\documentclass[fleqn,11pt]{article}
\usepackage{amscd,amsmath,amssymb,verbatim}
\usepackage{amsthm}
\usepackage[pdftex]{graphicx}
\usepackage{float}

\usepackage{authblk}

\usepackage{multirow,multicol}

\usepackage{subcaption}
\usepackage{booktabs}
\usepackage{bm}

\usepackage{siunitx}
\usepackage{adjustbox}

\usepackage[T1]{fontenc}
\usepackage{array}
\usepackage{mathrsfs}
\usepackage[nohead]{geometry}
\usepackage[singlespacing]{setspace}
\usepackage{rotating}
\usepackage{pdflscape}
\usepackage[table,xcdraw]{xcolor}

\usepackage[authoryear]{natbib}

\usepackage{hyperref}

\newtheorem{theorem}{Theorem}
\newtheorem{lemma}{Lemma}
\newtheorem{proposition}{Proposition}

\hoffset 
\voffset

\date{\today}
\title{A Simple and Robust Multi-Fidelity Data Fusion Method for Effective Modelling of Citizen-Science Air Pollution Data}

\author[1]{Camilla Andreozzi}
\author[2]{Pietro Colombo}
\author[2]{Philipp Otto\footnote{Corresponding author: philipp.otto@glasgow.ac.uk}}

\affil[1]{ETH Zurich}
\affil[2]{University of Glasgow, School of Mathematics and Statistics}

\begin{document}

\maketitle

\abstract{We propose a robust multi-fidelity Gaussian process for integrating sparse, high-quality reference monitors with dense but noisy citizen-science sensors. The approach replaces the Gaussian log-likelihood in the high-fidelity channel with a global Huber loss applied to precision-weighted residuals, yielding bounded influence on all parameters, including the cross-fidelity coupling, while retaining the flexibility of co-kriging. We establish attenuation and unbounded influence of the Gaussian maximum likelihood estimator under low-fidelity contamination and derive explicit finite bounds for the proposed estimator that clarify how whitening and mean-shift sensitivity determine robustness. Monte Carlo experiments with controlled contamination show that the robust estimator maintains stable MAE and RMSE as anomaly magnitude and frequency increase, whereas the Gaussian MLE deteriorates rapidly. In an empirical study of PM\(_{2.5}\) concentrations in Hamburg, combining UBA monitors with openSenseMap data, the method consistently improves cross-validated predictive accuracy and yields coherent uncertainty maps without relying on auxiliary covariates. The framework remains computationally scalable through diagonal or low-rank whitening and is fully reproducible with publicly available code.}

\section{Introduction}\label{sec:intro}

Reliable quantification of urban air pollution is essential for environmental assessment and policy design\footnote{See, for example, the UNDP blog on citizen engagement with low-cost sensors: \url{https://www.undp.org/policy-centre/singapore/blog/applications-low-cost-air-quality-sensors-citizen-engagement-and-air-pollution-mitigation}.}. In many cities, the availability of official high-precision monitoring sites remains limited, while extensive low-cost sensor networks provide complementary yet heterogeneous information. In Germany, for instance, the \emph{openSenseMap} platform \citep{pfeil2018opensensemap} enables citizen scientists to collect and share air-quality measurements using low-cost sensors that could complement official stations operated by the German Environment Agency (UBA). To address the resulting data quality and coverage challenges, recent research has increasingly adopted data-fusion strategies that integrate citizen observations with satellite retrievals and reference monitors to reduce bias and enhance predictive reliability. These approaches span both complex machine-learning architectures \citep{ghahremanloo2021deep,fu2023machine} and probabilistic frameworks designed to improve uncertainty quantification \citep{malings2024air}. Although such multi-fidelity data hold great potential for improving spatiotemporal coverage, they also introduce new statistical challenges: low-cost sensors often exhibit higher noise levels, change points, or sporadic outliers, and their error structure is rarely well characterised. Typical examples of these anomalies are illustrated in Figure~\ref{fig:lf_anomalies}, which shows PM$_{2.5}$ series from selected openSenseMap stations in Hamburg with abrupt level shifts and isolated extreme spikes. When the number of low-fidelity sensors grows into the tens or hundreds, systematic pre-filtering and manual quality control of individual time series become increasingly impractical. In particular, in an online monitoring context, where data streams are updated continuously and must be assimilated in real time, manual anomaly detection or sensor-specific calibration becomes infeasible. Consequently, conventional Gaussian process (GP) fusion models, which assume jointly Gaussian residuals and homogeneous noise, can become unstable or misleading when trained on contaminated observations.

Another aspect, more broadly related to spatio-temporal prediction, concerns the heavy reliance of many models on exogenous covariates. As a result, their predictive performance often does not generalize to regions where such predictors are unavailable or sparse. This limitation—frequently overlooked in spatio-temporal studies—fundamentally constrains the generalizability and practical applicability of new modeling approaches.

\begin{figure}
    \centering
    \begin{subfigure}[b]{0.45\textwidth}
        \centering
        \includegraphics[width=\textwidth]{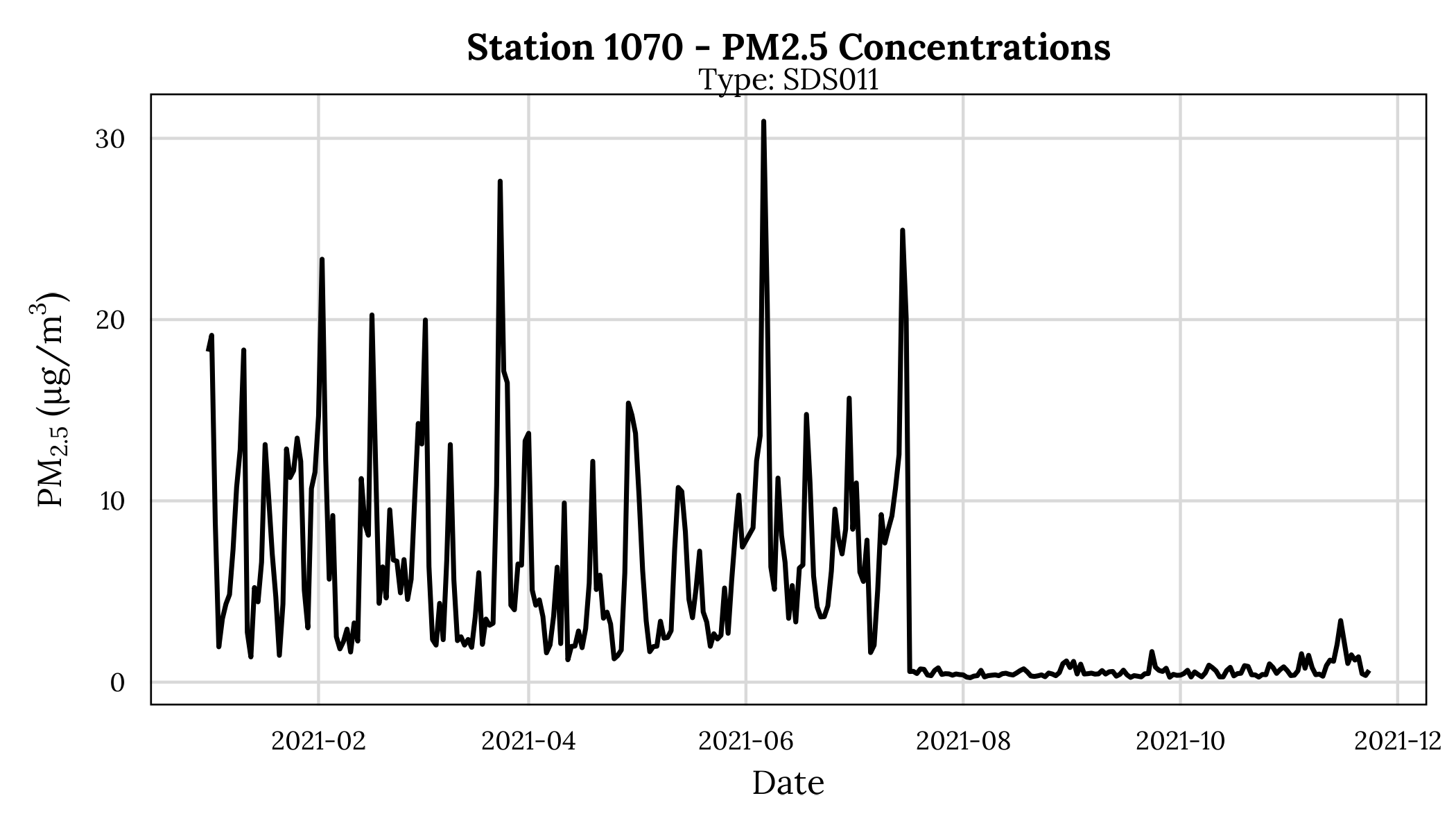}
        \caption{Station 1070}
    \end{subfigure}
    \hfill
    \begin{subfigure}[b]{0.45\textwidth}
        \centering
        \includegraphics[width=\textwidth]{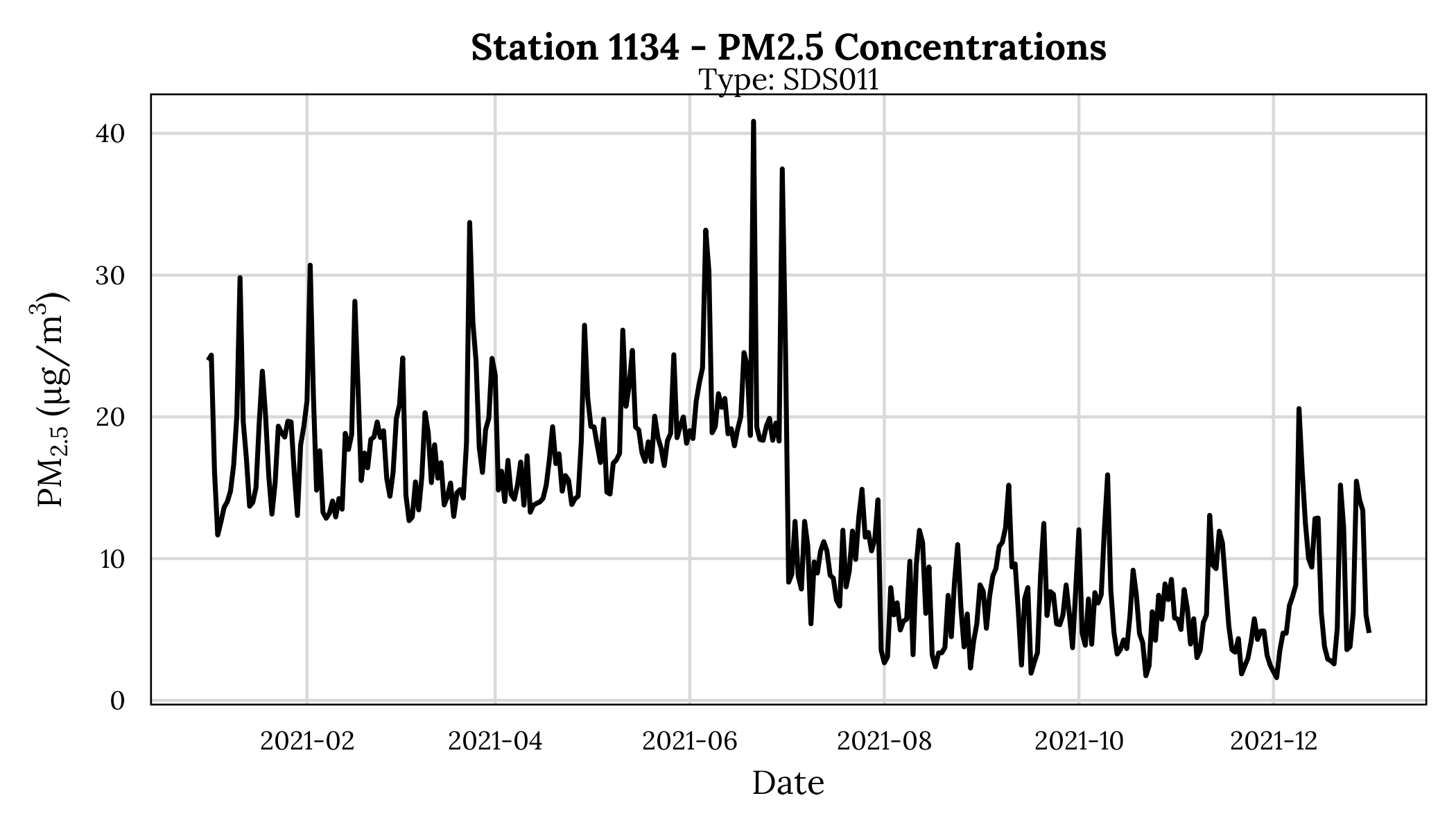}
        \caption{Station 1134}
    \end{subfigure}
    
    \vskip\baselineskip
    
    \begin{subfigure}[b]{0.45\textwidth}
        \centering
        \includegraphics[width=\textwidth]{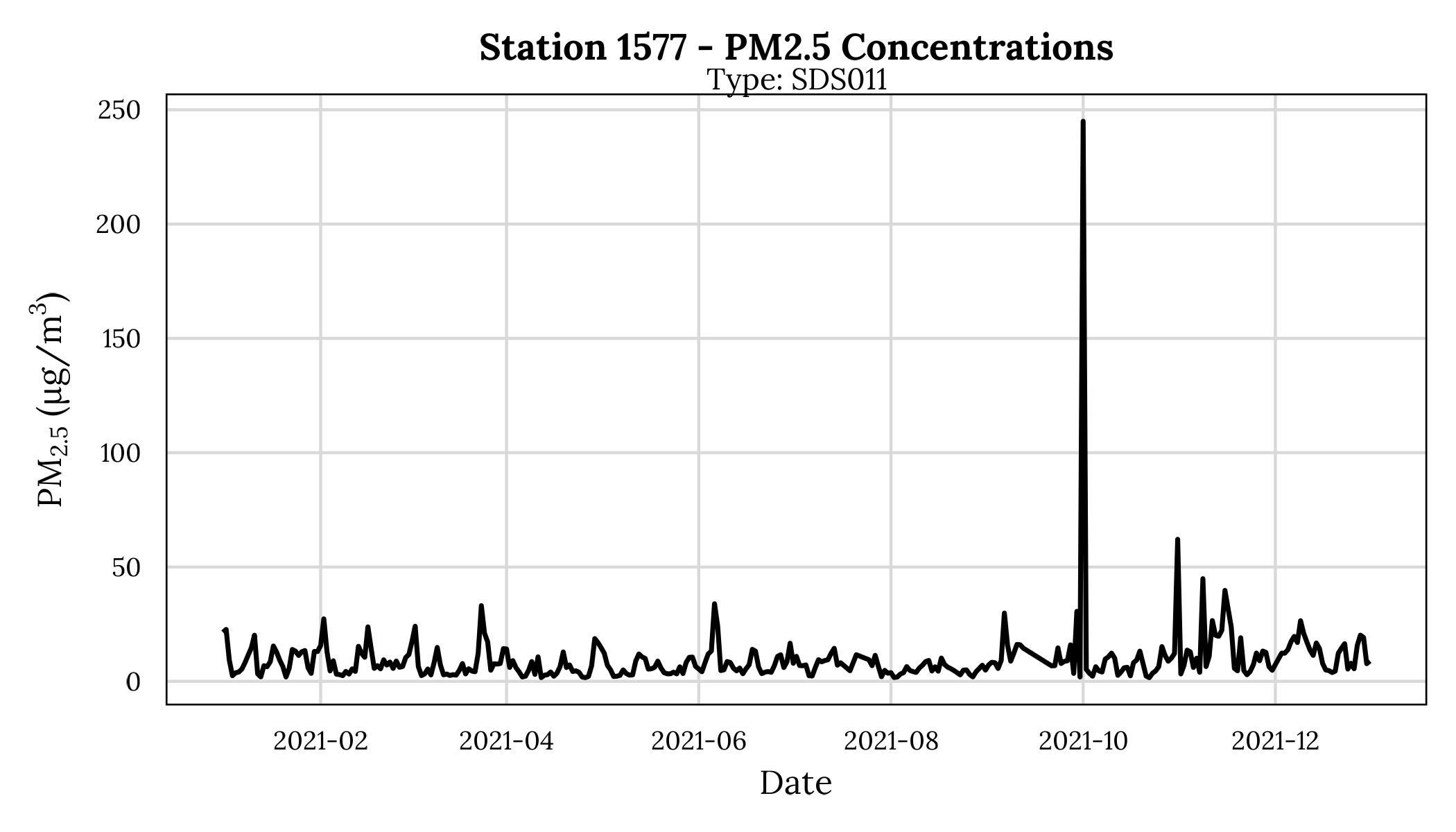}
        \caption{Station 1577}
    \end{subfigure}
    \hfill
    \begin{subfigure}[b]{0.45\textwidth}
        \centering
        \includegraphics[width=\textwidth]{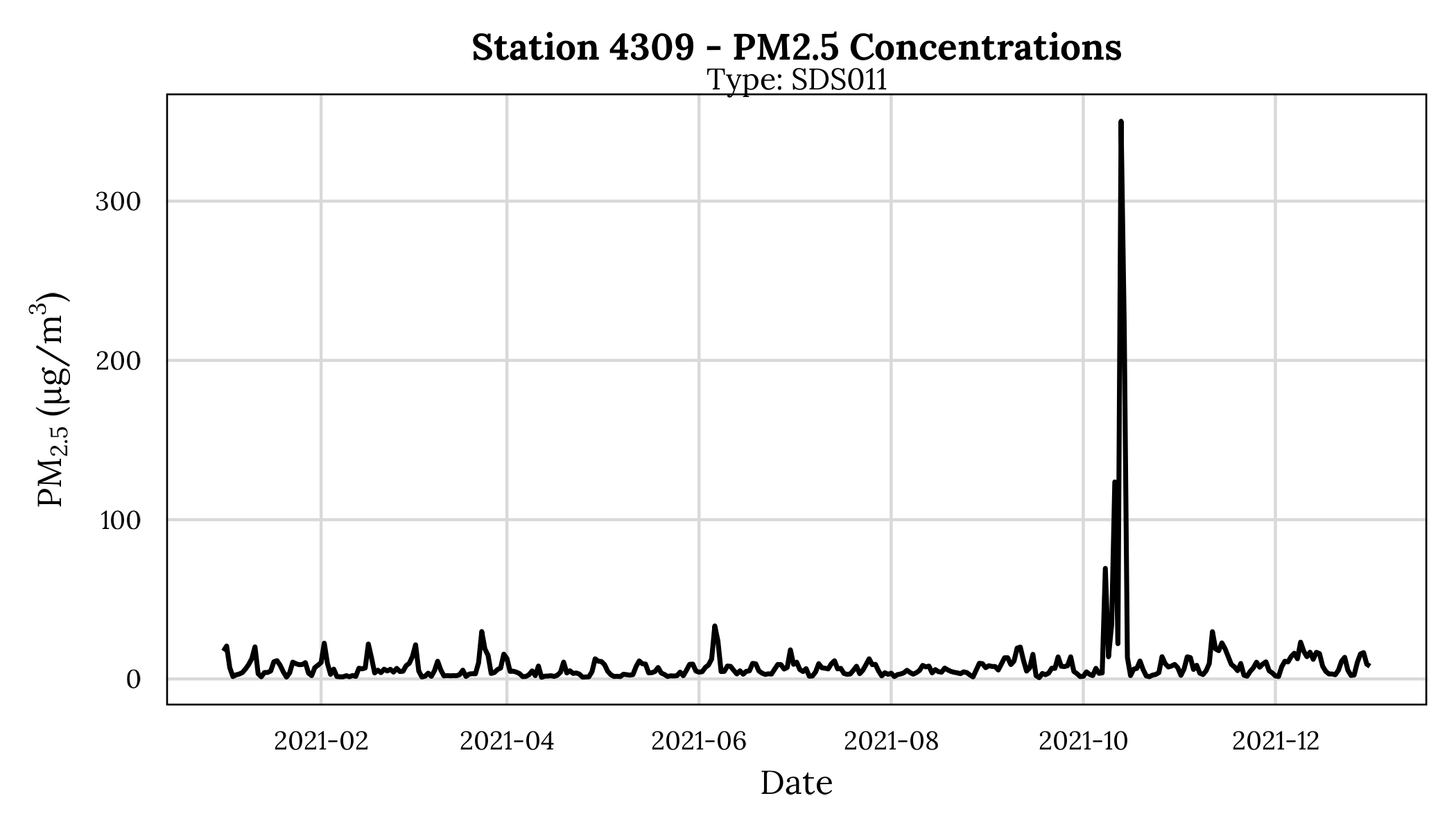}
        \caption{Station 4309}
    \end{subfigure}
    
    \caption{Examples of anomalous low-fidelity PM$_{2.5}$ time series from the openSenseMap citizen-sensor network in Hamburg.}
    \label{fig:lf_anomalies}
\end{figure}

Among the various spatio-temporal frameworks for data fusion, multi-fidelity modelling \citep{babaee2020multifidelity} %(add your paper with Fabio Pietro) 
based on a hierarchical ordering of data sources has demonstrated strong performance even without relying on external covariates. The statistical foundations of multi-fidelity modelling were laid by \citet{KennedyOHagan2000}, who introduced the autoregressive co-kriging framework to combine outputs from models of differing accuracy within a hierarchical Bayesian structure for code calibration and prediction. Subsequent contributions, such as the tutorial by \citet{OHagan2006}, consolidated the theoretical basis for Bayesian emulation and uncertainty quantification. Building on these ideas, \citet{Forrester2007} demonstrated the efficiency of multi-fidelity optimisation in engineering design, while \citet{LeGratiet2013} established the recursive co-kriging formulation that improved computational scalability and enabled systematic cross-validation.  More recent advances have extended the framework toward nonlinear and data-driven settings. For instance, \citet{Perdikaris2015} proposed a practical recursive implementation for large-scale engineering applications, and \citet{PerdikarisPNAS2017} generalised the classical linear autoregressive relation into a nonlinear information-fusion paradigm,  representing a major conceptual shift toward flexible and data-efficient multi-fidelity learning. More modern developments \citep{colombo2025warped} introduce warped multi-fidelity models that allow coherent normalisation of both fidelity levels.

To address the above-mentioned challenges, we propose the first Robust Multi-Fidelity Gaussian Process framework (RMFGP) that preserves the flexibility of classical co-kriging while maintaining stability under anomalies in low-fidelity data. Existing robust GP approaches mainly rely on heavy-tailed likelihoods \citep{jylanki2011robust} but are limited to single-output regression. Conversely, multi-fidelity Gaussian processes are highly sensitive to contamination but lack robust counterparts. To our knowledge, no prior model offers bounded-influence multi-fidelity co-kriging. Our proposed approach replaces the Gaussian log-likelihood by a global Huber loss, yielding a bounded-influence M-estimator for all model parameters, including the cross-fidelity correlation. We provide theoretical guarantees for bounded influence under both sparse and block-wise contamination, derive practical expressions for the influence bound, and show how a precision-weighted formulation ensures scalability to high-dimensional spatiotemporal settings.

Statistical air-quality modelling has evolved globally through a wide range of regional applications that combine methodological innovation with policy relevance. In North America, multilevel spatiotemporal frameworks integrating satellite, ground-based, and meteorological information have been extensively applied to ozone and particulate matter modelling, including long-term, high-resolution reconstructions for California \citep{Su2024} and comparative regional analyses for the U.S. Northeast \citep{Yanosky2014}. In East Asia, particularly in the Beijing--Tianjin--Hebei corridor, recent studies have employed multi-source data fusion and geographically weighted regression to capture severe pollution episodes and spatial nonstationarity in PM$_{2.5}$ dynamics \citep{li2023exploring}. Within Europe, the Lombardy region in northern Italy has emerged as a benchmark area, as it is the most polluted area in Europe. Using the harmonised Agrimonia dataset \citep{fasso2023agrimonia}, \citet{otto2024spatiotemporal} systematically compared different statistical and machine-learning models for PM$_{2.5}$ concentrations. Further studies have examined seasonal variability in pollutant levels \citep{colombo2024pm2}, structural changes in air quality \citep{maranzano2024spatiotemporal}, or the health impact of increased PM$_{10}$ concentrations \citep{renna2024impacts}. \citet{rodeschini2024scenario} demonstrated that variance varies significantly across space and time, motivating the use of a heteroscedastic spatiotemporal model for their scenario analyses. Notably, \citet{oxoli2024non} integrated non-conventional data sources for predicting PM$_{2.5}$ concentrations in Lombardy.

Empirically, the method is evaluated in two complementary stages. First, controlled Monte Carlo experiments quantify how outliers affect parameter estimates and the predictive accuracy. Second, the model is applied to {PM$_{2.5}$ concentrations in Hamburg}, Germany, integrating high-quality measurements with dense citizen-sensor observations from \emph{openSenseMap}. Importantly, the analysis relies solely on the intrinsic spatio-temporal structure of the data, without incorporating auxiliary covariates such as meteorological or satellite variables. This demonstrates that robust statistical modelling alone can effectively reconcile heterogeneous data fidelities, thereby ensuring high generalizability of the proposed method. The results show that the proposed estimator consistently outperforms benchmark and competitor models across most evaluation periods, while producing spatially coherent annual pollution maps and reliable uncertainty quantification even in the presence of uncleaned sensor data, as shown in Figure \ref{fig:lf_anomalies}.

% Beyond this application, the framework contributes a general methodology for \emph{robust multi-fidelity learning in environmental monitoring}, where data quality varies systematically across sources. It highlights how influence-function analysis and robust likelihood design can enhance reliability in spatiotemporal fusion, a theme increasingly relevant to environmental data science and statistical applications at the interface of machine learning and uncertainty quantification.
The remainder of the paper is structured as follows. Section~\ref{sec:methods} introduces the multi-fidelity Gaussian process formulation and establishes its theoretical properties under contamination, including the asymptotic bias and influence function of the Gaussian estimator. Section~\ref{sec:robust} develops the global Huber-loss extension and derives explicit bounds for the influence of low-fidelity anomalies. Section~\ref{sec:MC} presents a Monte Carlo study that validates the theoretical results under controlled perturbations. Section~\ref{sec:empirical} applies the methodology to real spatiotemporal air-quality data from Hamburg, demonstrating the benefits of robust multi-fidelity fusion for heterogeneous monitoring networks. Finally, Section~\ref{sec:conclusion} concludes the paper.

\section{Methods}\label{sec:methods}

Multi-fidelity Gaussian process (GP) models extend the classical single-fidelity GP regression framework to situations where data of different quality, resolution, or computational cost are available. The key idea is to exploit the information contained in \emph{low-fidelity} (LF) data while correcting for systematic discrepancies with respect to \emph{high-fidelity} (HF) data. This idea was first formalised by \citet{KennedyOHagan2000}, who proposed a hierarchical GP formulation that has since become the canonical framework for multi-fidelity modelling.

\subsection{Hierarchical multi-fidelity model}

Let \( f_L(\bm{x}) \) and \( f_H(\bm{x}) \) denote two stochastic processes defined on the input domain \( \mathcal{X} \subset \mathbb{R}^d \), representing the LF and HF responses, respectively. In spatiotemporal settings, \(\bm{x} = (\bm{s},t)\) typically combines spatial coordinates \(\bm{s} \in \mathbb{R}^2\) and time \(t \in \mathbb{R}\), allowing \(f_L\) and \(f_H\) to describe processes evolving continuously across both space and time. In our context of air quality monitoring, for example, the HF process represents measurements from a network of high-precision official stations using gravimetric reference instruments, while the LF process originates from dense citizen-science networks relying on optical sensors of lower accuracy but higher spatial coverage. Following \citet{KennedyOHagan2000}, the HF process is modelled as a linear transformation of the LF process, including a residual discrepancy process $\delta(\bm{x})$, i.e.,
\begin{equation}
    f_H(\bm{x}) = \rho\, f_L(\bm{x}) + \delta(\bm{x}),
    \label{eq:mf_model}
\end{equation}
where \( \rho \in \mathbb{R} \) is a scaling parameter that adjusts the overall magnitude and correlation structure between the two fidelities.

The LF process \( f_L(\bm{x}) \) is assumed to follow a Gaussian process
\begin{equation}
    f_L(\bm{x}) \sim \mathcal{GP}\!\big(\mu_L(\bm{x}),\, k_L(\bm{x}, \bm{x}'; \theta_L)\big),
\end{equation}
with mean function \( \mu_L(\cdot) \) and covariance function \( k_L(\cdot,\cdot) \) parametrised by \( \theta_L \). The structure of the covariance function can be adapted to the empirical setting and may, for instance, be chosen as stationary and isotropic in purely spatial applications, anisotropic when directional dependencies are expected, or separable and non-separable in the spatiotemporal case, depending on whether spatial and temporal correlations factorise \citep[see, e.g.,][for a general overview]{cressie2011statistics}. For our setting, we will use the squared exponential (or radial basis function, RBF) kernel,
\[
k(\bm{x}, \bm{x}') = \sigma_s^2 \exp\!\left(-\frac{\|\bm{x} - \bm{x}'\|^2}{2 l^2}\right),
\]
where \( \sigma_s^2 \) denotes the signal variance and \( l \) is the length-scale (or range) parameter controlling the rate at which correlations decay with increasing distance between locations. Further, the discrepancy process \( \delta(\bm{x}) \) is modelled as an independent Gaussian process 
\begin{equation}
    \delta(\bm{x}) \sim \mathcal{GP}\!\big(\mu_\delta(\bm{x}),\, k_\delta(\bm{x}, \bm{x}'; \theta_\delta)\big),
\end{equation}
representing the systematic difference between the two fidelities.

Combining the two levels, we obtain the joint process
\[
    \mathbf{f}(\bm{x}x) = 
    \begin{bmatrix}
        f_L(\bm{x}) \\[2pt]
        f_H(\bm{x})
    \end{bmatrix}
\]
as a multivariate Gaussian,
\begin{equation}
    \mathbf{f}(\bm{x}) \sim 
    \mathcal{GP}\!\left(
        \boldsymbol{\mu}(\bm{x}),
        \boldsymbol{K}(\bm{x},\bm{x}')
    \right),
\end{equation}
with the mean vector 
\[
    \boldsymbol{\mu}(\bm{x}) =
    \begin{bmatrix}
        \mu_L(\bm{x}) \\[2pt]
        \rho\,\mu_L(\bm{x}) + \mu_\delta(\bm{x})
    \end{bmatrix},
\]
and block covariance function
\begin{equation}
    \boldsymbol{K}(\bm{x},\bm{x}') =
    \begin{bmatrix}
        k_L(\bm{x},\bm{x}') & \rho\,k_L(\bm{x},\bm{x}') \\[2pt]
        \rho\,k_L(\bm{x},\bm{x}') & \rho^2\,k_L(\bm{x},\bm{x}') + k_\delta(\bm{x},\bm{x}')
    \end{bmatrix}.
    \label{eq:joint_kernel}
\end{equation}
Thus, the cross-covariance between the fidelities is determined by the scaling parameter \( \rho \), which encodes the degree of linear dependence between LF and HF processes. When \( |\rho| \) is close to one, the two fidelities are strongly correlated; smaller values of \( |\rho| \) indicate weaker dependence or structural mismatch. Notice that \( |\rho| \) is unbounded, hence it can also act as scale-multiplier of the overall signal LF. In practice, all parameters of the model, including the scaling factor \( \rho \), the covariance parameters \( \theta_L \) and \( \theta_\delta \), and the noise variances, must be estimated from the data, e.g., by maximising the Gaussian likelihood. However, the resulting estimates can be sensitive to irregularities in the observed LF data. Outliers, abrupt level shifts, or structural changes in the LF process may distort the joint likelihood and bias the estimation of \( \rho \) and the covariance components. Such irregularities are not uncommon in empirical applications, particularly when LF data originate from heterogeneous or low-cost sensor networks, motivating the development of robust estimation procedures for multi-fidelity models.

\subsection{Observation model and Gaussian likelihood estimation}

In practice, both processes are observed at finite sets of spatial locations over time. Therefore, let \( \{\bm{x}_i^L\}_{i=1}^{n_L} \) and \( \{\bm{x}_j^H\}_{j=1}^{n_H} \) denote the locations of the LF and HF data, respectively, and define
\[
    \mathbf{y}_L = f_L(\bm{x}_L) + \boldsymbol{\varepsilon}_L, 
    \qquad 
    \mathbf{y}_H = f_H(\bm{x}_H) + \boldsymbol{\varepsilon}_H,
\]
where \( \boldsymbol{\varepsilon}_L \sim \mathcal{N}(\mathbf{0}, \tau_L^2 I) \) and \( \boldsymbol{\varepsilon}_H \sim \mathcal{N}(\mathbf{0}, \tau_H^2 I) \) denote independent and identically distributed measurement errors, normally assumed to be Gaussian. Stacking both fidelities gives
\begin{equation}
    \mathbf{y} = 
    \begin{bmatrix}
        \mathbf{y}_L \\[2pt]
        \mathbf{y}_H
    \end{bmatrix}
    \sim
    \mathcal{N}\!\left(
    \begin{bmatrix}
        \boldsymbol{\mu}_L \\
        \rho\,\boldsymbol{\mu}_L + \boldsymbol{\mu}_\delta
    \end{bmatrix},
    \begin{bmatrix}
        K_{LL} + \tau_L^2 I & \rho K_{LL} \\
        \rho K_{LL} & \rho^2 K_{LL} + K_\delta + \tau_H^2 I
    \end{bmatrix}
    \right),
    \label{eq:joint_obs}
\end{equation}
where \(K_{LL}\) and \(K_\delta\) are the covariance matrices obtained from \(k_L\) and \(k_\delta\) at the corresponding sets of input locations. Under this joint Gaussian model, the parameters
\[
    \Theta = \{\rho, \theta_L, \theta_\delta, \tau_L^2, \tau_H^2\}
\]
are estimated by maximising the Gaussian log-likelihood
\begin{equation}
    \ell(\Theta) = -\tfrac{1}{2}\log |\Sigma(\Theta)|
    - \tfrac{1}{2}(\mathbf{y} - \boldsymbol{\mu}(\Theta))^\top 
    \Sigma(\Theta)^{-1}(\mathbf{y} - \boldsymbol{\mu}(\Theta)),
    \label{eq:loglik}
\end{equation}
where \( \Sigma(\Theta) \) is the block covariance matrix in \eqref{eq:joint_obs}.

The model formulation in \eqref{eq:mf_model}--\eqref{eq:joint_obs} provides the basis for subsequent theoretical developments. In particular, the parameter \( \rho \) is of primary interest, as it quantifies the linear coupling between fidelities. Below, we analyse the statistical behaviour of the classical maximum likelihood estimator (MLE) under the standard Gaussian assumptions and under deviations caused by contaminated or misspecified low-fidelity data.

In empirical multi-fidelity applications, the low-fidelity (LF) data source is often subject to irregularities that violate the idealised Gaussian model assumptions. Such distortions can arise from measurement errors (e.g., the PM sensor was positioned next to a BBQ area), local environmental effects (e.g., the sensor is placed in a corner with varying ventilation), or gradual sensor degradation (e.g., general wear). We will consider two principal types of contamination that are particularly relevant in spatiotemporal monitoring systems:

\begin{enumerate}
    \item \textbf{Outliers.} 
    Individual observations or short subsequences of the LF process may deviate markedly from the underlying trend due to transient sensor faults or unmodelled external influences. These are typically sparse in time or space and can be represented as additive contamination,
    \[
        y_L(\bm{x}_i) = f_L(\bm{x}_i) + u_i, \qquad 
        u_i \sim (1 - \eta)\,\mathcal{N}(0, \tau_L^2) + \eta\,\mathcal{Q},
    \]
    where a small fraction \( \eta \) of observations follow a heavy-tailed or shifted distribution \( \mathcal{Q} \). 

    \item \textbf{Change points or level shifts.}
    Longer-term deviations may occur when the LF process exhibits abrupt changes in its mean level or variance, for instance, after a change of the sensor position. 
    This can be expressed as
    \[
        f_L(\bm{x},t) =
        \begin{cases}
            f_L(\bm{x},t), & t \leq \tau,\\[3pt]
            f_L(\bm{x},t) + \Delta, & t > \tau,
        \end{cases}
    \]
    where \( \tau \) denotes the change-point location and \( \Delta \) represents the mean shift. 
    Such structural breaks lead to nonstationarity in the LF process and can bias both the scaling parameter \( \rho \) and the covariance hyperparameters when fitted under the stationary Gaussian model.
\end{enumerate}

In the following paragraphs, we formalise the impact of these deviations on the Gaussian MLE for the scaling parameter \( \rho \) and demonstrate analytically how outliers and level shifts lead to biased or unstable estimates. Let the high-, and low-fidelity residuals be denoted by
\[
\mathbf{r}_L=\mathbf{y}_L-\boldsymbol{\mu}_L,\qquad 
\mathbf{r}_H=\mathbf{y}_H-\boldsymbol{\mu}_H, \quad \text{and} \quad 
B = K_{LL}\Sigma_{LL}^{-1},\quad \Omega = K_\delta+\tau_H^2 I,
\]
so that, under the model in \eqref{eq:joint_obs}, the conditional relation is
\begin{equation}
\mathbf{r}_H=\rho\,B\,\mathbf{r}_L+\boldsymbol{\eta},\qquad \boldsymbol{\eta}\sim\mathcal{N}(\mathbf{0},\Omega).
\label{eq:cond_relation}
\end{equation}

\begin{proposition}[Standard closed form Gaussian MLE for $\rho$]
\label{prop:gls}
For fixed $(\theta_L,\theta_\delta,\tau_L^2,\tau_H^2)$, the Gaussian MLE of $\rho$ equals the generalised least squares (GLS) estimator
\begin{equation}
\hat\rho
=\frac{\mathbf{r}_L^\top B^\top \Omega^{-1}\mathbf{r}_H}{\mathbf{r}_L^\top B^\top \Omega^{-1} B\,\mathbf{r}_L}.
\label{eq:rho_hat}
\end{equation}
\end{proposition}

This follows from the fact that, conditional on $\mathbf{y}_L$, the Gaussian log-likelihood in $\rho$ reduces to a quadratic form in $\mathbf{r}_H-\rho B\mathbf{r}_L$, so that maximising it is equivalent to solving the GLS normal equations \citep[see, e.g.,][]{rasmussen2006gaussian}. Now, assume the LF observations are contaminated:
\begin{equation}
\mathbf{y}_L=\mathbf{f}_L+\mathbf{u},\qquad \mathbb{E}[\mathbf{u}]=\mathbf{0},\quad \mathrm{Cov}(\mathbf{u})=\Sigma_u, \quad \mathbf{u}\ \perp\ (\mathbf{f}_L,\boldsymbol{\eta}), 
\label{eq:contam_setup}
\end{equation}
where the shift vector $\mathbf{u}$ arises from the two different shift scenarios discussed above. In the following theorem, we show that the estimated scale parameter $\rho$ is biased. 
% although it is asymptotically bounded.

\begin{theorem}[Pseudo-true parameter under LF contamination]
\label{thm:attenuation}
Assume that the contamination is independent of the latent LF process $\mathbf{f}_L$ and the HF innovations $\boldsymbol{\eta}$, i.e, $\mathbf{u}\perp(\mathbf{f}_L,\boldsymbol{\eta})$, and let $\rho^\star$ denote the Kullback--Leibler projection of the Gaussian MLE under \eqref{eq:contam_setup}. Then,
\begin{equation}
\rho^\star
=\rho\,
\frac{\mathrm{tr}\!\big(B^\top \Omega^{-1} C_L B\big)}
{\mathrm{tr}\!\big(B^\top \Omega^{-1} (C_L+\Sigma_u) B\big)}
\ =:\ \kappa\,\rho,\qquad \kappa\in(0,1],
\label{eq:rho_star}
\end{equation}
where $C_L=\mathrm{Cov}(\mathbf{f}_L)$. That is, any nonzero contamination variance $\Sigma_u$ leads to a biased estimate and attenuates $\rho^\star$ toward $0$.
\end{theorem}

The proof of all theoretical results obtained throughout the paper can be found in the Appendix. Even though the asymptotic bias of $\rho^\star$ is bounded, as $|\rho-\rho^\star|= (1-\kappa)|\rho|\le |\rho|$, the influence of extreme LF contamination on the finite-sample MLE is unbounded. To show that, let the score function for $\rho$ from \eqref{eq:cond_relation}:
\begin{equation}
s_\rho(\mathbf{y})=\frac{\partial \ell(\rho)}{\partial \rho}
=\big(\mathbf{r}_H-\rho B\mathbf{r}_L\big)^\top \Omega^{-1} B\,\mathbf{r}_L, \quad \text{and}
\label{eq:score}
\end{equation}
let $A=-\mathbb{E}\big[ \partial s_\rho/\partial\rho \big]=\mathbb{E}\big[\mathbf{r}_L^\top B^\top\Omega^{-1}B\,\mathbf{r}_L\big]>0$. The influence function at the uncontaminated model is
\begin{equation}
\mathrm{IF}(\mathbf{y};\rho)= -\,A^{-1}\,s_\rho(\mathbf{y}).
\label{eq:IF_def}
\end{equation}
In the following two propositions, we will demonstrate that the two change types will result in an unbounded influence function.

\begin{proposition}[Point outliers imply unbounded influence]
\label{prop:IF_outlier}
Fix $\mathbf{r}_H$ and $\rho\neq0$. For any sequence of low-fidelity contaminations with $\|\mathbf{u}\|\to\infty$,
\[
|s_\rho(\mathbf{y})|\to\infty
\quad\text{and hence}\quad
|\mathrm{IF}(\mathbf{y};\rho)|\to\infty.
\]
Thus, the Gaussian MLE has an unbounded influence with respect to LF outliers.
\end{proposition}

\begin{proposition}[Level shifts imply unbounded influence]
\label{prop:IF_shift}
Let $\mathbf{u}=\mathbf{b}$ represent a deterministic mean shift in the LF process, 
with $b_i=\Delta$ on a nonempty index set $\mathcal{I}$ and $b_i=0$ otherwise. 
Fix $\mathbf{r}_H$ and $\rho\neq 0$. 
As $|\Delta|\to\infty$,
\[
|s_\rho(\mathbf{y})|\to\infty
\quad\text{and hence}\quad
|\mathrm{IF}(\mathbf{y};\rho)|\to\infty.
\]
Therefore, the Gaussian MLE exhibits unbounded influence with respect to the magnitude of an unmodelled mean shift in the LF process.
\end{proposition}

These theoretical results show that the Gaussian MLE converges to an attenuated pseudo-true value (bounded asymptotic bias). Yet, its influence function is unbounded for both sparse outliers and block-wise level shifts, implying that finite-sample distortions can be arbitrarily large as contamination magnitude increases or as a few LF points attain extreme leverage. Thus, even a few contaminated LF observations can induce arbitrarily large distortions in~$\hat\rho$
and, through the conditional dependence $f_H(\bm{x})=\rho f_L(\bm{x})+\delta(\bm{x})$, propagate these distortions to the high-fidelity predictions.

There are generally two ways to robustify this multifidelity setting. First, at the level of the cross-fidelity relation~\eqref{eq:cond_relation}, the Gaussian score~\eqref{eq:score} can be replaced by its Huber-weighted counterpart, where large standardised LF residuals receive reduced weights. This local approach directly limits the leverage of anomalous LF data on~$\hat\rho$. The resulting estimator~$\hat\rho_H$ only has a bounded influence, ensuring that even extreme LF outliers or level shifts cannot cause unbounded perturbations in~$\hat\rho$. Since the HF conditional mean depends linearly on~$\mathbf{r}_L$, bounded perturbations in~$B W \mathbf{r}_L$ translate into bounded perturbations of the HF predictor~$\hat{\mathbf{y}}_H$. Alternatively, second, robustness can therefore be enforced globally by replacing the Gaussian log-likelihood with a Huber loss on the standardised HF residuals. This global objective down-weights large deviations in~$\mathbf{r}_H$, regardless of whether they stem from genuine HF anomalies or from propagated LF contamination via the linear coupling $f_H(\bm{x})=\rho f_L(\bm{x})+\delta(\bm{x})$. In this formulation, the Huber loss acts as a global influence limiter across fidelities, yielding stable joint estimation of $(\theta_L,\theta_\delta,\rho)$ even under mixed contamination.

\subsection{Global robustification via Huber weighting}\label{sec:robust}

The Gaussian log-likelihood in~\eqref{eq:loglik} treats all residuals as equally reliable and is therefore sensitive to extreme deviations in either fidelity. To improve stability under contamination, we replace the quadratic log-likelihood loss with a \emph{Huber loss} applied to the standardised high-fidelity residuals. Because the high-fidelity process satisfies $f_H(\bm{x})=\rho f_L(\bm{x})+\delta(\bm{x})$, any distortion in the low-fidelity data propagates linearly to the high-fidelity residuals. Down-weighting large high-fidelity residuals thus mitigates the effect of both direct HF anomalies and LF contamination.

Let 
\[
\mathbf{r}_H = \mathbf{y}_H - \boldsymbol{\mu}_H(\Theta)
\quad \text{with} \quad
\boldsymbol{\mu}_H(\Theta) = \rho\,\boldsymbol{\mu}_L + \boldsymbol{\mu}_\delta,
\]
denote the residuals of the high-fidelity model at locations
$\mathbf{x}_H=\{x_j^H\}_{j=1}^{n_H}$.
Let
\[
W_{HH}(\Theta)
=\big(\rho^2 K_{LL} + K_\delta + \tau_H^2 I_{n_H}\big)^{-1}
\]
denote the precision matrix corresponding to the high-fidelity block of $\Sigma(\Theta)$ in~\eqref{eq:joint_obs}. The residuals are standardised (whitened) using $W_{HH}^{1/2}$,
\[
\tilde{\mathbf{r}}_H = W_{HH}^{1/2}\mathbf{r}_H,
\]
which yields approximately homoscedastic unit-variance components under the ideal Gaussian model.

To robustify the estimation, we replace the quadratic term in the Gaussian log-likelihood~\eqref{eq:loglik} by a \emph{Huber loss} applied to the whitened high-fidelity residuals~$\tilde{\mathbf{r}}_H=W_{HH}^{1/2}\mathbf{r}_H$. This modification preserves the Gaussian form for small residuals but down-weights the contribution of large deviations, thereby limiting their influence on parameter estimation. The Huber loss function is defined as
\begin{equation}
    \rho_\delta(r) =
\begin{cases}
\tfrac{1}{2}r^2, & |r|\le \delta,\\[2pt]
\delta(|r|-\tfrac{1}{2}\delta), & |r|>\delta,
\end{cases}
\label{eq:Huber_obj}
\end{equation}

where the transition constant~$\delta$ controls the boundary between the quadratic and linear regimes. Its value is determined from a robust scale estimate of the standardised residuals using the median absolute deviation,
\[
\hat{s} = \frac{\operatorname{median}_i(|\tilde{r}_{H,i}|)}{0.6745}, 
\qquad 
\delta = 1.345\,\hat{s}.
\]
Define the global Huber estimator as the minimiser of the robust objective \eqref{eq:Huber_obj}:
\begin{equation}
\widehat{\Theta}_H \;\in\;
\arg\min_{\Theta}\;
\mathcal{H}(\Theta)
\;=\;
\arg\min_{\Theta }
\sum_{i=1}^{n_H}
\rho_\delta\!\Big(
\big(W_{HH}^{1/2}(\Theta)\,[\,\mathbf{y}_H-\boldsymbol{\mu}_H(\Theta)\,]\big)_i
\Big),
\label{eq:ThetaHatHuber}
\end{equation}
where
\(
W_{HH}(\Theta)
=\big(\rho^2 K_{LL} + K_\delta + \tau_H^2 I_{n_H}\big)^{-1}
\)
is the HF precision block from \eqref{eq:joint_obs}, and \(\boldsymbol{\mu}_H(\Theta)=\rho\,\boldsymbol{\mu}_L+\boldsymbol{\mu}_\delta\). Let
\(
\tilde{\mathbf r}_H(\Theta)=W_{HH}^{1/2}(\Theta)\,[\,\mathbf{y}_H-\boldsymbol{\mu}_H(\Theta)\,]
\)
and \(\psi_\delta=\rho'_\delta\). In the following theorem, we demonstrate that this global robustification results in a bounded influence of both shift types discussed above.

\begin{theorem}[Bounded influence of the global Huber estimator]
\label{thm:bounded_IF_global}
Let $\widehat{\Theta}_H$ denote the estimator given by~\eqref{eq:ThetaHatHuber}, and let $\mathrm{IF}_H(\mathbf{y})$ denote its influence function. Assume that the low-fidelity process may be contaminated by either (i) sparse outliers $\mathbf{u}$ with $\|\mathbf{u}\|\to\infty$ at finitely many indices, or (ii) block-wise mean shifts $\mathbf{u}=\mathbf{b}$ with $b_i=\Delta$ on a subset $\mathcal{I}$ and $b_i=0$ otherwise, with $|\Delta|\to\infty$. Define $\tilde{\mathbf r}_H(\Theta)
= W_{HH}^{1/2}(\Theta)\,[\mathbf y_H-\boldsymbol{\mu}_H(\Theta)]$ and the Huber score
$\psi_\delta=\rho'_\delta$ with fixed $\delta>0$. Let
\[
S(\Theta;\mathbf y)=\sum_{i=1}^{n_H}\psi_\delta(\tilde r_{H,i}(\Theta))\,g_i(\Theta),
\qquad
g_i(\Theta):=\frac{\partial \tilde r_{H,i}(\Theta)}{\partial \Theta^\top},
\]
and suppose the Jacobian
$J(\Theta):=\mathbb E\big[\partial S(\Theta;\mathbf Y)/\partial \Theta^\top\big]\big|_{\Theta=\Theta_0}$
exists at the ideal model and is nonsingular.

Then, for any fixed $\delta>0$,
\[
\sup_{\mathbf y}\ \|\mathrm{IF}_H(\mathbf y)\| \;\le\; C_\delta \;<\;\infty,
\]
with the explicit bound
\begin{equation}
C_\delta \;\le\; \|J^{-1}\|\;\delta\;\sum_{i=1}^{n_H}\|g_i(\Theta)\|,
\label{eq:Cdelta-general}
\end{equation}
where $\|\cdot\|$ denotes the standard Euclidean norm for vectors and its induced operator norm for matrices.
\end{theorem}

Theorem~\ref{thm:bounded_IF_global} establishes that the global Huber-based objective yields an estimator with uniformly bounded local influence, irrespective of the magnitude or form of contamination in the LF process. This extends standard robustness theory to the multi-fidelity Gaussian-process framework, where distortions in the low-fidelity data propagate indirectly through the linear relation in~\eqref{eq:cond_relation} into the high-fidelity likelihood. The robustness arises from the structure of the Huber score $\psi_\delta$, whose derivative remains capped at~$\pm\delta$, limiting the contribution of each standardised residual $\tilde r_{H,i}$ and thereby preventing unbounded leverage effects even under extreme contamination.

The constants in the bound~\eqref{eq:Cdelta-general} provide further insight into how robustness is achieved. The factor $\|J^{-1}\|$ quantifies the curvature and identifiability of the robust objective in a neighbourhood of the true parameter $\Theta_0$: under standard regularity conditions, $J$ is nonsingular and $\|J^{-1}\|$ remains finite, guaranteeing local stability of the estimator. The additional factors in~\eqref{eq:Cdelta-general} decompose the influence constant~$C_\delta$ into interpretable components, revealing how the conditioning of the whitening operator $W_{HH}^{1/2}$ and the sensitivity of the high-fidelity mean $\boldsymbol{\mu}_H(\Theta)$ jointly determine the attainable degree of robustness. A well-conditioned whitening matrix and moderate gradient norms reduce the propagation of contamination effects and thus lower~$C_\delta$.

Because $|\psi_\delta(r)|\le\delta$ for all residuals, the contamination magnitude itself does not appear in the bound. Hence, the global Huber loss ensures finite influence for both isolated outliers and block-wise level shifts---an essential contrast to the unbounded distortions of the Gaussian maximum-likelihood estimator established in Propositions~\ref{prop:IF_outlier} and \ref{prop:IF_shift}. This property confirms that global robustification provides a safeguard against extreme or systematic anomalies in multi-fidelity spatiotemporal data.

To make these dependencies explicit in practice, we distinguish two computational approaches that can be followed in robust multi-fidelity estimation. The first corresponds to the general case in which $W_{HH}^{1/2}$ depends on~$\Theta$, as when the covariance matrix of the high-fidelity process is re-estimated in each optimisation step. Here, the bound must capture the joint sensitivity of both $W_{HH}^{1/2}$ and $\boldsymbol{\mu}_H(\Theta)$, leading to the Lipschitz-type result in Lemma~\ref{lem:alt-bounds-general}. The second regime occurs in nested implementations, where $W_{HH}^{1/2}$ is held fixed while updating~$\Theta$. In this simplified setting, the dependence on the whitening operator disappears, and the bound reduces to the expression in Lemma~\ref{lem:alt-bounds-fixed}. This second approach is computationally much more efficient, at the cost of reduced robustification bounds. 

\begin{lemma}[General whitening]
\label{lem:alt-bounds-general}
Let $\tilde{\mathbf r}_H(\Theta)=W_{HH}^{1/2}(\Theta)\,[\mathbf y_H-\boldsymbol{\mu}_H(\Theta)]$ and
$g_i(\Theta):=\partial \tilde r_{H,i}(\Theta)/\partial \Theta^\top$.
Assume that the parameter space $\Xi$ of $\Theta$ is compact and $\Theta\in \Xi$, $W_{HH}^{1/2}(\Theta)$ is continuously differentiable on $\Xi$, and
\[
\sup_{\Theta\in\Xi}\Big\|\frac{\partial W_{HH}^{1/2}}{\partial \Theta}\Big\|_2 \le L_W,\quad
\sup_{\Theta\in\Xi}\Big\|\frac{\partial \boldsymbol{\mu}_H}{\partial \Theta}\Big\|_2 \le L_\mu,\quad
\sup_{\Theta\in\Xi}\|W_{HH}^{1/2}(\Theta)\|_2 \le \kappa_W,
\]
and set $R:=\sup_{\Theta\in\Xi}\|\mathbf r_H(\Theta)\|_2<\infty$, with $\mathbf r_H(\Theta)=\mathbf y_H-\boldsymbol{\mu}_H(\Theta)$.
Then
\begin{equation}
\sum_{i=1}^{n_H}\|g_i(\Theta)\| \ \le\ n_H\big(L_W\,R+\kappa_W\,L_\mu\big),
\qquad
C_\delta \ \le\ \|J^{-1}\|\,\delta\,n_H\big(L_W\,R+\kappa_W\,L_\mu\big),
\label{eq:Cdelta-Lipschitz-split}
\end{equation}
where $C_\delta$ is the constant in \eqref{eq:Cdelta-general} and $J$ is the population Jacobian from Theorem~\ref{thm:bounded_IF_global}. 
\end{lemma}

\begin{lemma}[Fixed whitening]
\label{lem:alt-bounds-fixed}
Suppose $W_{HH}^{1/2}$ is held fixed while updating $\Theta$.
Then
\[
g_i(\Theta)=-\,e_i^\top W_{HH}^{1/2}\,\frac{\partial \boldsymbol{\mu}_H}{\partial \Theta},
\quad
G(\Theta):=\begin{bmatrix}g_1(\Theta)^\top\\ \vdots \\ g_{n_H}(\Theta)^\top\end{bmatrix}
= -\,W_{HH}^{1/2}\,\frac{\partial \boldsymbol{\mu}_H}{\partial \Theta}.
\]
Let $||\cdot||$ denote the Frobenius norm. Consequently,
\begin{equation}
\sum_{i=1}^{n_H}\|g_i(\Theta)\|
\ \le\ \sqrt{n_H}\,\|G(\Theta)\|_F
\ \le\ \sqrt{n_H}\,\|W_{HH}^{1/2}\|
\Big\|\frac{\partial \boldsymbol{\mu}_H}{\partial \Theta}\Big\|_F,
\label{eq:sumgi-fixed}
\end{equation}
and the influence bound specialises to
\begin{equation}
C_\delta
\ \le\ \|J^{-1}\|\,\delta\,\sqrt{n_H}\,\|W_{HH}^{1/2}\|\,
\Big\|\frac{\partial \boldsymbol{\mu}_H}{\partial \Theta}\Big\|_F.
\label{eq:Cdelta-fixedW-split}
\end{equation}
\end{lemma}

Moreover, in practice, constructing the full whitening transformation $W_{HH}^{1/2}$ can be computationally demanding for large spatiotemporal datasets, as it involves either Cholesky or eigendecomposition of the high-fidelity covariance block $K_{HH}$. A common simplification---also adopted in our implementation---is to approximate $W_{HH}$ by its diagonal, and applying elementwise standardisation rather than full whitening. This strategy substantially reduces computational cost and memory requirements from $\mathcal{O}(n_H^3)$ to $\mathcal{O}(n_H)$, enabling scalable, robust inference even in high-resolution multi-fidelity settings. 

To assess the practical implications of these theoretical findings, we next conduct a Monte Carlo study comparing the Gaussian and Huber-robust estimators under controlled contamination scenarios. The simulations quantify the finite-sample behaviour of bias and variance in $\hat\rho$ and confirm the stabilising effect of the global Huber loss in both outlier and level-shift settings.

\section{Monte-Carlo simulation}\label{sec:MC}

To examine the empirical performance of the proposed robustification, we conducted a Monte Carlo study based on realistically structured spatiotemporal data. The simulated dataset consists of observations from 16 spatial stations (arranged on a regular $4\times4$ grid) recorded at 15 equally spaced time points over the unit interval (see Appendix~\ref{Simulation} for full details of the data-generation process). To emulate measurement anomalies, controlled perturbations were introduced into the low-fidelity observations, with both the \emph{frequency} (proportion of affected points) and the \emph{magnitude} (severity of deviation) systematically varied. The resulting design, summarised in Table~\ref{tab:perturbation_parameters}, comprises nine distinct contamination scenarios%, as summarised in Table \ref{tab:perturbation_parameters}%
. We performed 100 MC replications.

\begin{table}
\centering
\caption{Perturbation parameters used in the simulation experiment.}
\vspace{2mm}
\begin{tabular}{ccc}
\hline
\textbf{Parameter} & \textbf{Description} & \textbf{Values} \\
\hline
$m$ & Magnitude of the distortion & $\{2,\,5,\,10\}$ \\
$\eta$ & Frequency of the distortion & $\{0.1,\,0.3,\,0.5\}$ \\
$n_{\text{runs}}$ & Runs per scenario & $100$ \\
\hline
\end{tabular}
\label{tab:perturbation_parameters}
\end{table}

For each Monte Carlo replication, the spatial locations were randomly split into training and test sets, with 80\% of the stations used for model fitting and the remaining 20\% reserved for evaluation. The entire time series at the test stations was withheld to preserve the temporal dependence structure \citep[see][for a review on cross-validation in spatiotemporal statistics]{otto2024review}. Two spatiotemporal multi-fidelity models were then fitted to the training data using a squared-exponential covariance function with a separable, multiplicative space-time structure. The \emph{classical} model was estimated by minimising the negative log marginal likelihood in~\eqref{eq:loglik}, whereas the \emph{robust} model employed the Huber-loss estimator given by~\eqref{eq:ThetaHatHuber}. Both optimisation procedures used a quasi-Newton algorithm with identical random initialisation to ensure comparability. 

Model performance was evaluated at the high-fidelity test stations $\{\mathbf{x}_i\}_{i=1}^{n_{\text{test}}}$ using the mean absolute error (MAE) and root mean squared error (RMSE),
\[
\mathrm{MAE}=\frac{1}{n_{\text{test}}}\sum_{i=1}^{n_{\text{test}}}\bigl|\hat{y}_H(\mathbf{x}_i)-y_H(\mathbf{x}_i)\bigr|,\qquad
\mathrm{RMSE}=\sqrt{\frac{1}{n_{\text{test}}}\sum_{i=1}^{n_{\text{test}}}\bigl(\hat{y}_H(\mathbf{x}_i)-y_H(\mathbf{x}_i)\bigr)^2}.
\]
For each successful replication, we recorded the estimated hyperparameters, model predictions, and corresponding test targets, from which MAE and RMSE were computed and subsequently aggregated across all contamination scenarios. All simulation code and results are available in the accompanying GitHub repository: \url{https://github.com/Pietrostat193/Robust-Multi-fidelity-Modelling/tree/main/Simulation/MatlabSimulations/PreliminaryExperiment}.

Table~\ref{tab:robust_comparison} summarises the predictive performance of the classical and robust multi-fidelity estimators across contamination magnitudes~$m$ and frequencies~$\eta$. The results reveal a clear pattern: the performance of the classical Gaussian maximum-likelihood estimator deteriorates systematically with increasing contamination magnitude and frequency, whereas the robust estimator based on the global Huber loss remains largely stable across all scenarios.

For small and infrequent perturbations ($m=2$, $\eta=0.1$), both estimators perform comparably, indicating that the Huber loss does not substantially reduce efficiency in nearly Gaussian conditions. As the contamination becomes stronger or more frequent, the MLE rapidly loses stability--its MAE and RMSE increase almost linearly with~$m$--reflecting the unbounded influence of large deviations in the low-fidelity process. In contrast, the robust estimator maintains a nearly constant error level even under severe contamination ($m=10$, $\eta=0.5$), confirming the theoretical results of Section~\ref{sec:methods} that the Huber objective bounds the influence of anomalous observations.

A slight efficiency loss of the robust estimator can be observed in the weakest contamination scenario, where its MAE exceeds that of the Gaussian MLE. This behaviour is expected: in purely Gaussian settings, the MLE is asymptotically efficient, while the Huber estimator sacrifices a small amount of efficiency to gain robustness. The trade-off is negligible for practical purposes and becomes favourable as soon as mild contamination is present, where the robust approach consistently outperforms the classical one. To quantify this effect, we computed the relative efficiency of the robust estimator with respect to the MLE as
\[
\mathrm{Eff}_{\mathrm{rel}}=\frac{\text{RMSE}_{\text{MLE}}^2}{\text{RMSE}_{\text{Huber}}^2}.
\]
In the weakest contamination setting ($m=2$), the robust estimator shows a larger efficiency loss than we would expect. This behaviour likely reflects overly conservative weighting of the standardised residuals, caused by a tuning constant $\delta=1.345\,\hat{s}$ that is too small relative to the effective residual scale and/or inflation effects due to diagonal whitening by $W_{HH}^{1/2}$. Together with the natural tendency of the MLE to adapt more flexibly to small random fluctuations in nearly clean data, this behaviour explains why the classical estimator can appear more efficient in the weakest contamination scenario, even though it rapidly deteriorates as anomalies become stronger. In practice, we recommend the following two strategies. First, the tuning constant $\delta$ can be adapted as $\delta=c\,\hat{s}$ to target a nominal Gaussian efficiency (e.g., $e(c)\approx0.95$), either using known efficiency curves or a cross-validation on pre-cleaned data. This adjustment ensures that moderate Gaussian residuals remain within the quadratic region of the Huber loss. Second, the whitening matrix $W_{HH}$ can be stabilised by using a full (Cholesky-based) whitening, a regularised approaches, for instance via $W_{HH} = (K_{HH}+\lambda I)^{-1}$ with small $\lambda$ \citep[see][]{otto2024review}, or low-rank approximations of $W_{HH}^{1/2}$. These approaches can balance scalability and statistical efficiency, maintaining robustness while avoiding excessive variance inflation.

Figures~\ref{fig:Scenario1_example} and~\ref{fig:Scenario9_example} illustrate representative predictions for the weakest and strongest contamination scenarios, respectively. In Scenario~1 ($m=2$, $\eta=0.1$; Figure~\ref{fig:Scenario1_example}), both estimators closely follow the test series, with the classical model achieving slightly lower MAE and RMSE due to its greater flexibility in adapting to small random fluctuations. In contrast, under severe contamination (Scenario~9, $m=10$, $\eta=0.5$; Figure~\ref{fig:Scenario9_example}), the robust estimator remains stable and closely aligned with the true signal, while the classical model shows pronounced deviations around anomalous peaks, confirming the theoretical finding that its influence function is unbounded under large distortions.

Overall, these results confirm that the proposed global Huber-robustification effectively stabilises parameter estimation and prediction in spatiotemporal multi-fidelity Gaussian process models, preventing the propagation of outlier effects from low-fidelity data into high-fidelity predictions.

\begin{table}
\centering
\caption{Comparison of classical and robust multi-fidelity estimators across contamination magnitudes $m$ and frequencies $\eta$. The relative efficiency is computed as $\mathrm{Eff}_{\mathrm{rel}} = (\mathrm{RMSE}_{\text{Classic}} / \mathrm{RMSE}_{\text{Robust}})^2$.}
\label{tab:robust_comparison}
\vspace{2mm}
\begin{tabular}{
  c c
  S[table-format=1.3]
  S[table-format=1.3]
  S[table-format=1.3]
  S[table-format=1.3]
  S[table-format=1.2]
}
\toprule
 &  & \multicolumn{2}{c}{\textbf{Classic}} & \multicolumn{2}{c}{\textbf{Robust}} & \textbf{Rel.\ Eff.}\\
\cmidrule(lr){3-4} \cmidrule(lr){5-6}
\textbf{$m$} & \textbf{$\eta$} 
  & \textbf{MAE} & \textbf{RMSE} 
  & \textbf{MAE} & \textbf{RMSE} 
  & \textbf{$\mathrm{Eff}_{\mathrm{rel}}$} \\
\midrule
\multirow{3}{*}{2} 
 & 0.1 & 0.593 & 0.758 & 0.988 & 1.241 & 0.37 \\
 & 0.3 & 0.700 & 0.913 & 1.065 & 1.338 & 0.47 \\
 & 0.5 & 0.906 & 1.178 & 1.096 & 1.383 & 0.72 \\[2pt]
\multirow{3}{*}{5} 
 & 0.1 & 0.675 & 0.902 & 1.085 & 1.375 & 0.43 \\
 & 0.3 & 1.067 & 1.501 & 1.128 & 1.436 & 1.09 \\
 & 0.5 & 1.547 & 2.002 & 1.273 & 1.587 & 1.59 \\[2pt]
\multirow{3}{*}{10} 
 & 0.1 & 0.846 & 1.343 & 1.100 & 1.394 & 0.93 \\
 & 0.3 & 1.563 & 2.213 & 1.111 & 1.395 & 2.52 \\
 & 0.5 & 1.998 & 2.381 & 1.317 & 1.584 & 2.26 \\
\bottomrule
\end{tabular}
\end{table}

\begin{figure}
    \centering
    \includegraphics[width=1\linewidth]{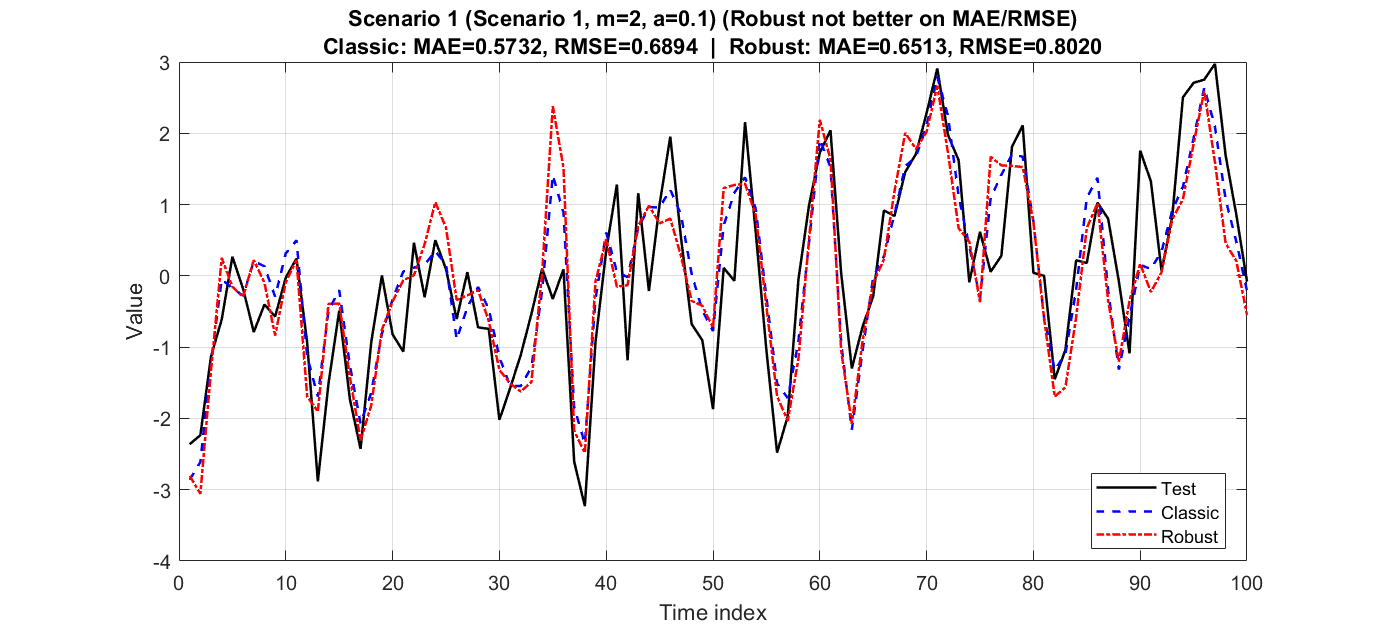}
    \caption{Illustrative prediction example from the simulation study under \textbf{Scenario~1} (low contamination).}
    \label{fig:Scenario1_example}
\end{figure}

\begin{figure}
    \centering
    \includegraphics[width=0.8\linewidth]{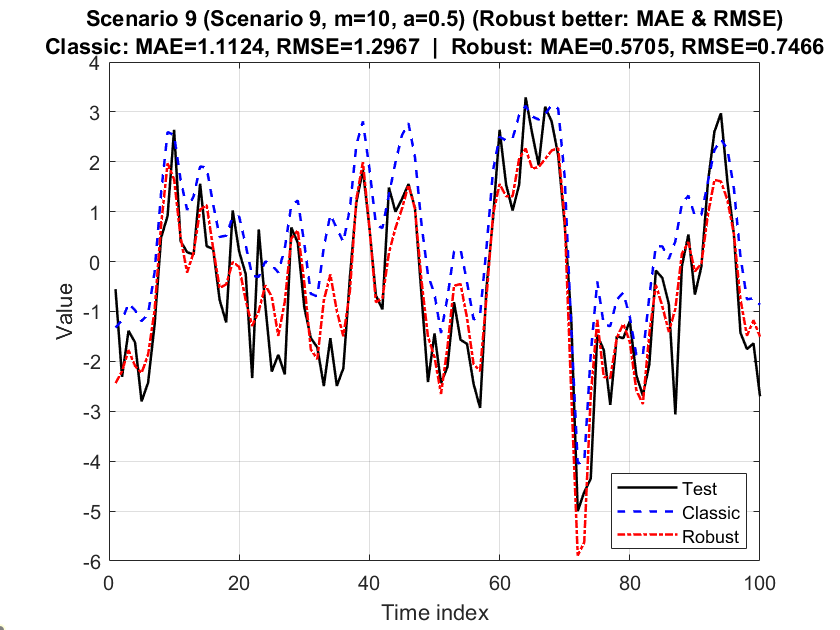}
    \caption{Illustrative prediction example from the simulation study under \textbf{Scenario~9} (high contamination).}
    \label{fig:Scenario9_example}
\end{figure}

\section{Empirical case study: multi-fidelity air quality fusion}\label{sec:empirical}

We now return to the air-quality case study introduced in Section~\ref{sec:intro}, where the objective is to integrate sparse, high-quality measurements from official monitoring stations with spatially dense but less accurate citizen-science data. We focus on the Hamburg metropolitan area, within the following coordinates $9.7^{\circ}-10.15^{\circ}E$ and $53.49^{\circ}-53.62^{\circ}N$. The high-fidelity (HF) data originate from four air quality monitoring stations operated by the German Environment Agency (UBA), based on gravimetric sensors, while the low-fidelity (LF) data stem from 40 low-cost optical sensors from \texttt{senseBox} platforms published through the \texttt{opensensemap.org}. We have taken the 15 nearest LF sites around each HF site, resulting in a total of 26 LF sensors. The openSenseMap platform, established in 2014, provides an open-access infrastructure for real-time environmental data collection and dissemination, based on a distributed network of community-operated sensors \citep{pfeil2018opensensemap}. As noted by \citet{bruch2025reviewing}, data quality in citizen-science sensor networks remains limited by inconsistent or incomplete user-provided metadata, often leading to semantic ambiguities and reduced interoperability. This heterogeneous sensing network provides an ideal testbed for assessing the robustness and predictive performance of proposed robust approaches. The locations of all measurement stations are shown in Figure \ref{fig:map}.

\begin{figure}
    \centering
    \includegraphics[width=1\linewidth]{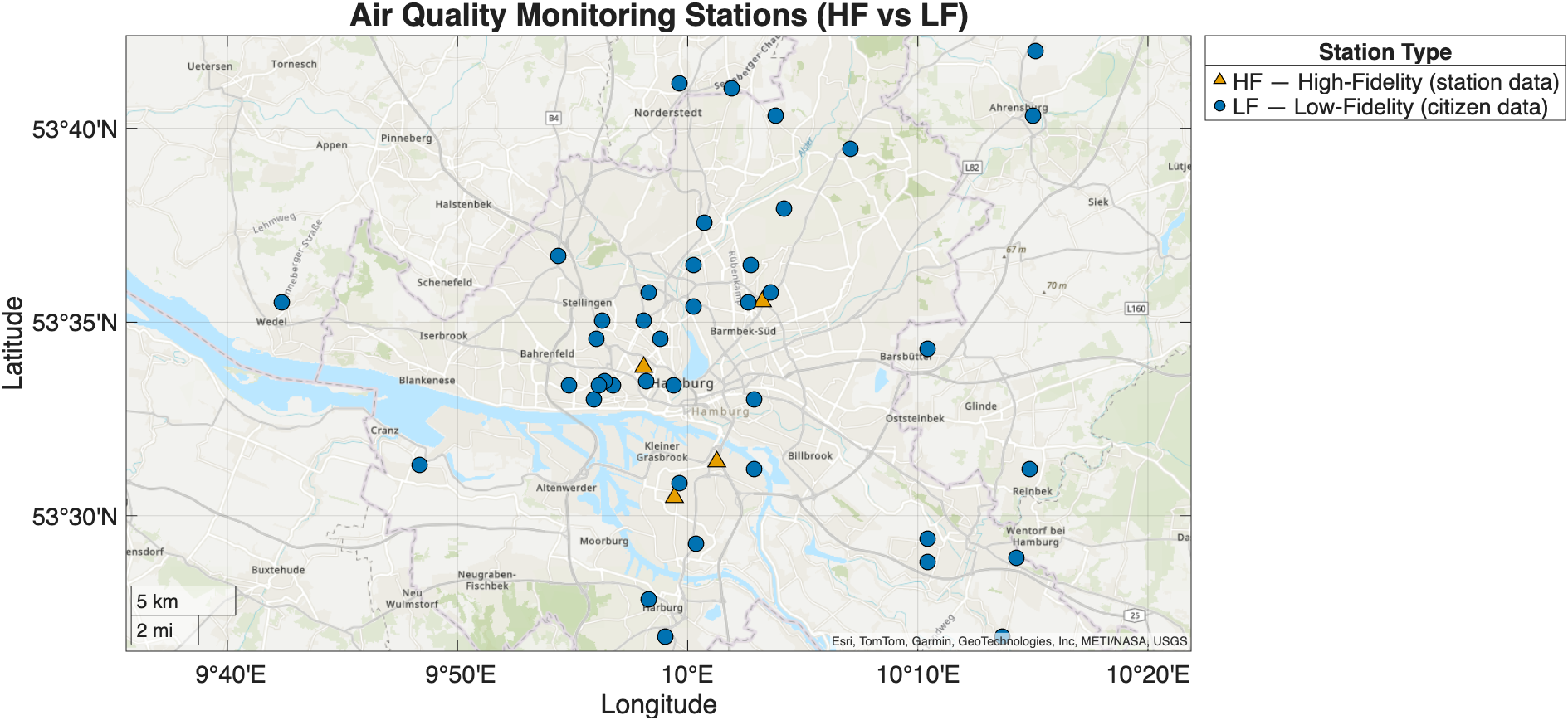}
    \caption{Map of sensor locations, colour-coded by sensor type (high-fidelity reference stations and low-cost citizen sensors).}
    \label{fig:map}
\end{figure}

Each station provides approximately 330 daily average $\text{PM}_{2.5}$ concentrations from January 1, 2021, to December 31, 2021, with slight variations across stations. This leads to a total number of observations of $\approx$13,000 for citizen-science sensors (SDS011) and $\approx$1,100 for the reference stations (STA). In Table \ref{tab:descriptive}, we provide a summary of descriptive statistics aggregated by the station type. The complete list of descriptive statistics for each sensor can be found in Table \ref{tab:descriptive_full} in the Appendix.

\begin{table}
\centering
\caption{PM\textsubscript{2.5} Summary by Sensor Type}\label{tab:descriptive}
\begin{adjustbox}{max width=\textwidth}
\begin{tabular}{lrrrrrrr}
\toprule
\textbf{Station ID} & \textbf{Count} & \textbf{Min} & \textbf{Max} & \textbf{Mean} & \textbf{Std.\ Error} & \textbf{95\% CI (Lower)} & \textbf{95\% CI (Upper)} \\
\midrule
SDS011 & 13180 & 0.082 & 999.900 & 9.326 & 0.1887 & 8.956 & 9.696 \\
STA & 1164 & 1.500 & 40.320 & 10.728 & 0.1963 & 10.343 & 11.113 \\
\bottomrule
\end{tabular}
\end{adjustbox}
\end{table}

To evaluate performance in both space and time, we employ a spatiotemporal block cross-validation scheme. The time series are partitioned into consecutive blocks of 30 time points (days), processed sequentially across the full observation period. Within each temporal block, one HF station is entirely held out for testing, while the remaining stations and all LF data are used for training. Iterating this procedure across the 11 temporal windows and four HF stations yields $11\times4=44$ distinct test sets. This combined temporal and spatial partitioning preserves temporal dependence within blocks, while providing a comprehensive assessment of predictive performance under realistic data availability conditions. We compare the performance of six competing models:
\begin{enumerate}
    \item the proposed robust multi-fidelity Gaussian process (RMFGP), which employs a global Huber loss;
    \item a well-calibrated RMFGP variant with adaptive tuning of the robustness parameter~$\delta$;
    \item the classical multi-fidelity Gaussian process (MFGP) estimated via maximum likelihood;
    \item a Bayesian spatio-temporal data-fusion model based on the INLA--SPDE framework~\citep{villejo2023data};
    \item a gradient-boosted regression tree model (QGBRT), representing a nonparametric machine-learning baseline for spatio-temporal prediction; and
    \item a robust single-fidelity Gaussian process, implemented in the \texttt{sdmTMB} package~\citep{anderson2022sdmtmb}.
\end{enumerate}

Together, these models span classical, robust, and alternative approaches to spatiotemporal data fusion. The comparison focuses on predictive accuracy, as measured by CV-MAE and CV-RMSE, of the HF stations, as shown in Table \ref{tab:emp_results}. We evaluate predictive
performance only at HF stations, as they provide the most reliable ground truth for assessing
the benefits of multi-fidelity fusion.

\begin{table}[ht]
\centering
\caption{Predictive performance (cross-validated MAE and RMSE) across temporal windows and competing models. Best values per window are shown in bold.}
\label{tab:emp_results}
\begin{adjustbox}{max width=\textwidth}
\begin{tabular}{c cccccc}
\toprule
\textbf{Time window} 
 & \textbf{MFGP} 
 & \textbf{RMFGP 1} 
 & \textbf{RMFGP 2} 
 & \textbf{INLA--SPDE} 
 & \textbf{QGBRT} 
 & \textbf{Robust SF} \\
\midrule
\multicolumn{7}{c}{\textbf{Cross-validated MAE}} \\
\midrule
1  & 4.80 & \textbf{1.43} & \textbf{1.43} & 6.86 & 5.07 & 2.41 \\
2  & 39.3 & \textbf{1.91} & \textbf{1.91} & 7.28 & 7.23 & 6.50 \\
3  & 18.9 & 11.2 & 11.2 & 6.37 & 7.05 & \textbf{3.96} \\
4  & 6.16 & \textbf{1.89} & \textbf{1.89} & 4.79 & 4.78 & 3.97 \\
5  & 8.49 & \textbf{2.35} & \textbf{2.35} & 3.25 & 3.94 & 3.18 \\
6  & 3.06 & \textbf{2.40} & \textbf{2.40} & 5.96 & 8.60 & 7.69 \\
7  & 1.61 & \textbf{1.32} & \textbf{1.32} & 2.48 & 5.39 & 5.35 \\
8  & 3.14 & \textbf{1.31} & \textbf{1.31} & 2.77 & 3.81 & 2.59 \\
9  & 2.12 & \textbf{2.11} & \textbf{2.11} & 4.40 & 4.63 & 4.72 \\
10 & 4.68 & \textbf{1.91} & \textbf{1.91} & 4.15 & 9.69 & 5.59 \\
11 & 8.06 & \textbf{3.45} & \textbf{3.45} & 5.83 & 9.00 & 8.13 \\
\midrule
\multicolumn{7}{c}{\textbf{Cross-validated RMSE}} \\
\midrule
1  & 5.80 & \textbf{1.59} & \textbf{1.59} & 9.01 & 6.10 & 3.16 \\
2  & 53.2 & \textbf{2.36} & 2.58 & 10.7 & 9.19 & 8.02 \\
3  & 24.3 & 14.2 & 14.1 & 10.4 & 8.12 & \textbf{4.30} \\
4  & 7.80 & \textbf{2.27} & \textbf{2.28} & 6.57 & 5.84 & 4.95 \\
5  & 10.2 & \textbf{2.83} & \textbf{2.83} & 3.91 & 4.58 & 3.68 \\
6  & 3.71 & \textbf{3.64} & \textbf{3.60} & 7.55 & 11.55 & 10.5 \\
7  & 1.70 & \textbf{1.36} & \textbf{1.36} & 3.34 & 6.74 & 5.85 \\
8  & 4.32 & \textbf{1.68} & \textbf{1.68} & 3.72 & 5.07 & 3.61 \\
9  & \textbf{2.56} & 2.61 & 2.61 & 5.25 & 5.51 & 5.85 \\
10 & 5.12 & \textbf{2.37} & \textbf{2.38} & 4.94 & 14.81 & 6.82 \\
11 & 8.50 & \textbf{4.09} & \textbf{4.09} & 7.02 & 10.76 & 9.74 \\
\bottomrule
\end{tabular}

\end{adjustbox}
\end{table}

The results in Table~\ref{tab:emp_results} show a clear and consistent performance gain of the proposed robust multi-fidelity Gaussian process models compared with both the classical multi-fidelity MLE and the single-fidelity alternatives. Across almost all temporal windows, the two Huber-based variants (\emph{RMFGP~1} and \emph{RMFGP~2}) achieve the smallest MAE and RMSEs. The classical multi-fidelity model occasionally collapses under heavy local deviations, particularly in windows 2 and 3, where unmodeled bias or scale shifts in the low-cost sensor data lead to large residual propagation into the high-fidelity component. In contrast, the robust estimators maintain stable prediction accuracy across time, with MAE values typically below~2--3~$\mu\mathrm{g/m}^3$. Figure~\ref{fig:mae_temporal} summarises the cross-validated MAE across temporal and spatial partitions. The classical multi-fidelity estimator shows pronounced fluctuations in performance, with clear degradation in specific time windows, likely corresponding to periods affected by local contamination or transient shifts in the low-fidelity data. In contrast, both Huber-based variants remain relatively stable across time and space, showing only mild variability between folds. This spatiotemporal stability illustrates the ability of the robust estimators to mitigate the influence of anomalous LF inputs without compromising accuracy during cleaner intervals.

Among the benchmarks, the Bayesian INLA--SPDE and the robust single-fidelity model perform competitively in some periods but are generally less accurate, likely reflecting their weaker coupling structure between fidelities. Notably, the robust single-fidelity model performs significantly better for the third temporal window, compared to all other approaches. The gradient-boosted regression tree (QGBRT) exhibits substantially higher errors, suggesting that the non-parametric strategy is unable to effectively capture the multiscale temporal dependencies present in the data. Overall, the global Huber multi-fidelity estimator strikes the best balance between efficiency and robustness, combining the interpretability of a Gaussian process framework with resilience against anomalous or systematically biased low-fidelity observations.

Figure~\ref{fig:kriging} illustrates the spatiotemporal kriging abilities of the approach. We display the spatial distribution of the time-averaged $\mathrm{PM}_{2.5}$ concentration in the $11^{th}$ time window and its associated predictive uncertainty obtained from the robust multi-fidelity GP model. The left panel shows the kriging-based mean, which shows elevated concentration levels in the urban areas, Hamburg harbour, Hamburg airport, and the populated city centre areas. Notably, green areas in the south-east, around the protected nature reserve areas (Kirchwerder Wiesen), are showing lower $\mathrm{PM}_{2.5}$ concentrations. The right panel displays the corresponding standard deviation field, which reflects the spatial pattern of information density: uncertainty increases in peripheral regions, where less sensors are available and high-fidelity sensors, in particular.

\begin{figure}
    \centering
    \includegraphics[width=1\linewidth]{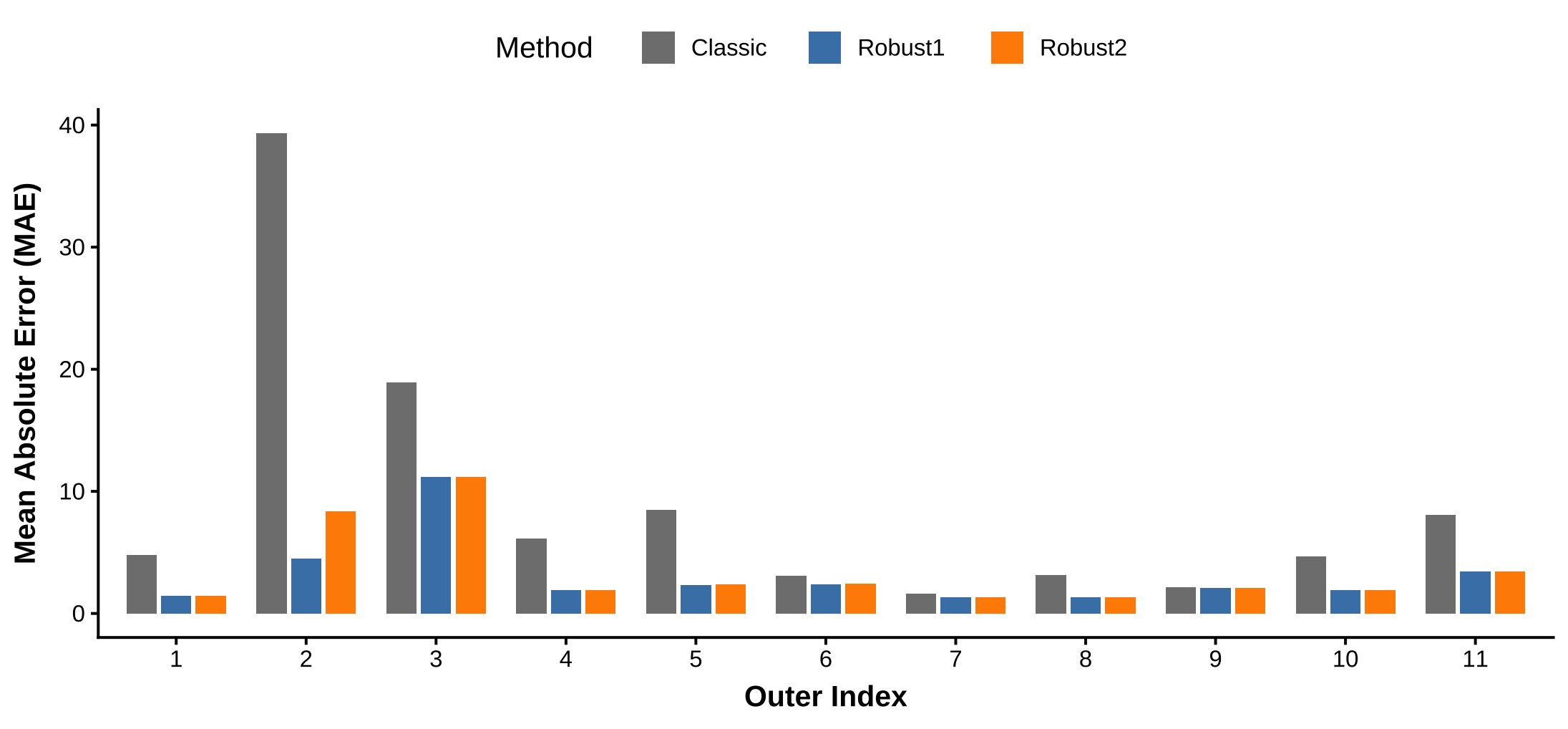}
    \includegraphics[width=1\linewidth]{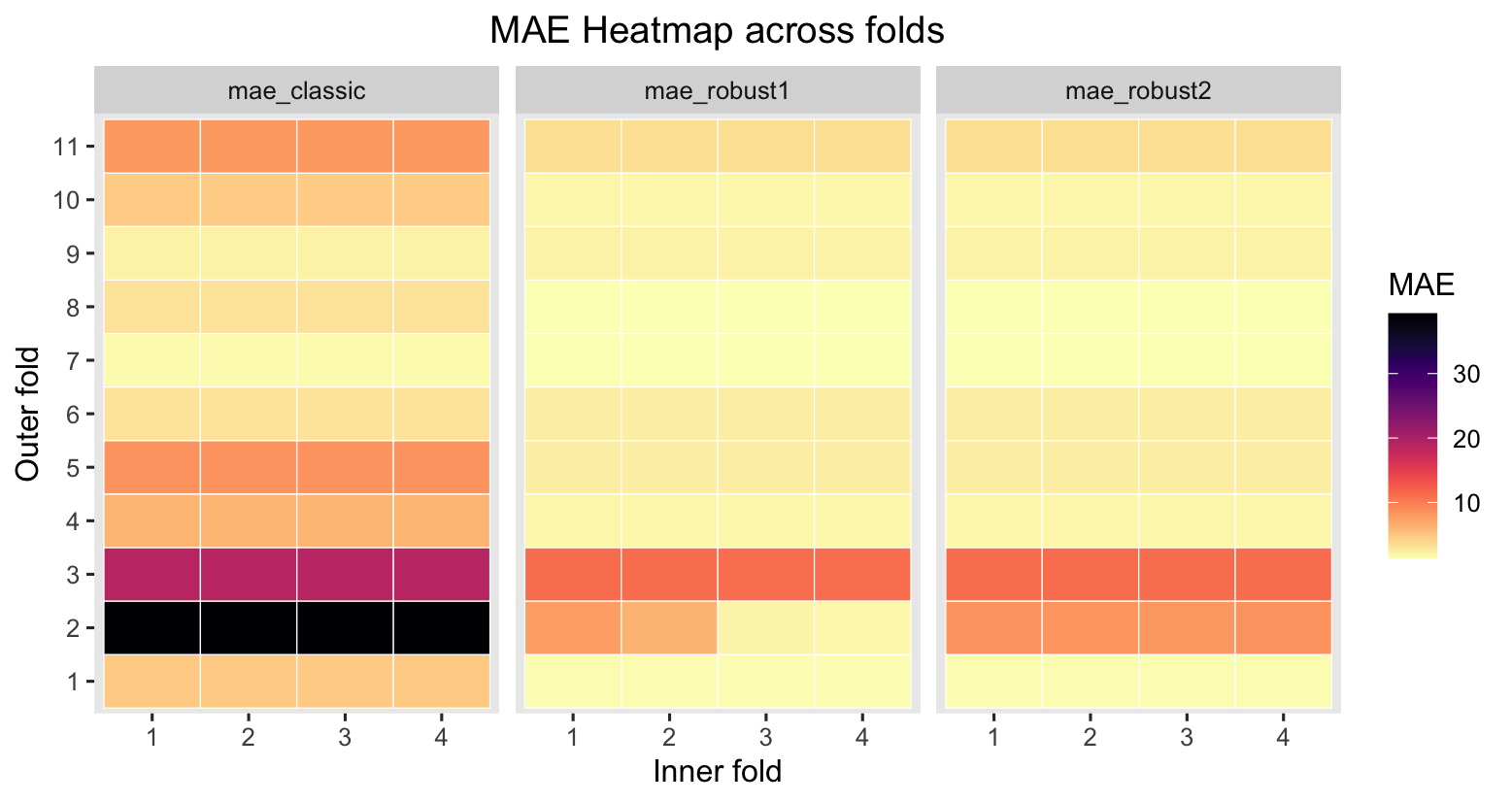}
    \caption{Comparison of mean absolute error (MAE) across temporal and spatial validation folds. The top panel shows the MAE for each temporal window (outer fold), while the bottom heatmaps display the joint variation across temporal (outer) and spatial (inner) folds.}
    \label{fig:mae_temporal}
\end{figure}

\begin{figure}
    \centering
    \includegraphics[width=1\linewidth]{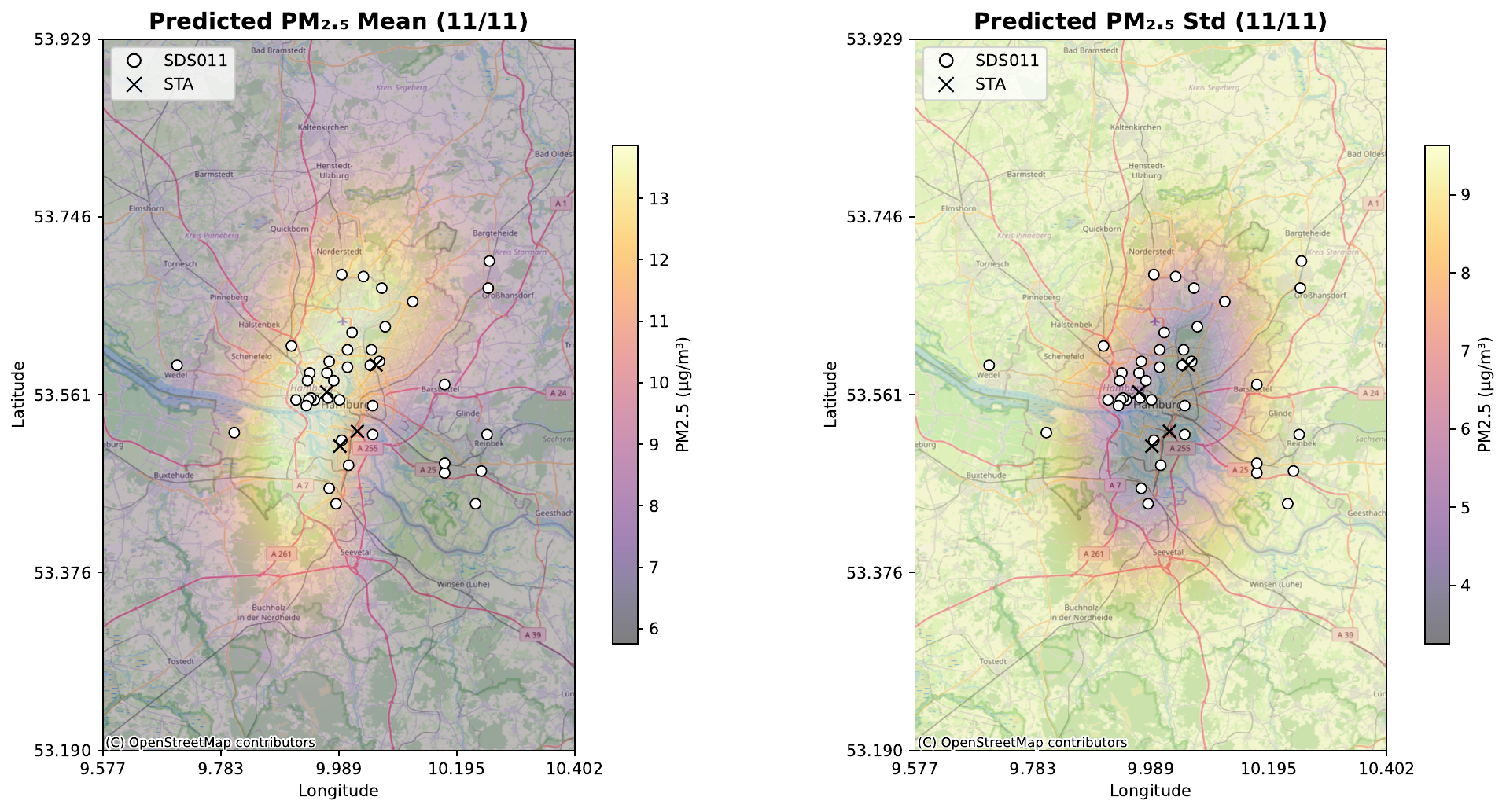}
    \caption{Spatial distribution of predicted PM${2.5}$ concentrations ($\mu$g/m³) in the Hamburg metropolitan area for the 11th estimation time window (winter period). Predictions are generated using the robust multi-fidelity Gaussian Process (GP) model described in this paper. The (left) panel displays the temporal mean PM${2.5}$ concentration across the time window. The (right) panel displays the corresponding temporal standard deviation, indicating the variability of concentrations within that window. Points indicate the locations of the monitoring stations by type used to train the model.}
    \label{fig:kriging}
\end{figure}

\section{Conclusion}\label{sec:conclusion}

Reliable data fusion between heterogeneous information sources is essential in modern environmental and engineering applications, where high-fidelity reference measurements are complemented by abundant but noisy low-cost observations, e.g., from unvalidated citizen-science measurements. However, such multi-fidelity systems are intrinsically vulnerable to contamination and systematic bias in the low-fidelity data, which can severely distort classical Gaussian process estimators. This paper addressed this challenge by developing a globally robust spatiotemporal multi-fidelity Gaussian process model that replaces the standard log-likelihood with a Huber-type loss applied to precision-weighted high-fidelity residuals. 

The proposed framework extends classical multi-fidelity co-kriging to accommodate outliers and level shifts while preserving interpretability and scalability. Theoretical analysis established that the global Huber estimator has a bounded influence function for both sparse and block-wise contamination, contrasting with the unbounded sensitivity of the Gaussian maximum likelihood estimator. Explicit bounds were derived, showing how curvature, whitening, and parameter sensitivity jointly determine robustness. 

Simulation experiments demonstrated that the robust model remains stable under increasing contamination magnitude, while the classical estimator deteriorates sharply. In a real-world case study combining official UBA monitoring stations and openSenseMap citizen-sensor data, the robust MFGP substantially improved prediction accuracy and temporal consistency, highlighting its ability to reconcile quality and coverage in heterogeneous sensor networks.  Future research could explore adaptive calibration of the robustness parameter, scalable covariance approximations, and integration for high-frequency and real-time environmental monitoring. Overall, the proposed approach provides a statistically grounded and computationally feasible extension of multi-fidelity Gaussian processes, bridging the gap between robust statistics and modern data-fusion applications.

% \bibliographystyle{apalike}
% \bibliography{bibliography} % Name of your .bib file (without .bib)

\section*{Acknowledgment}

AI Disclosure: ChatGPT (OpenAI, GPT-4 and GPT-5.1, November 2025) was used for language editing and drafting selected paragraphs. All model outputs were critically reviewed and edited by the authors to ensure accuracy, clarity, and alignment with the scientific content. The use of the tool did not extend to data analysis, model development, or interpretation of results.

\section{Appendix}

\subsection{Data generation process}\label{sec:dgp}

\noindent
\hangindent=-5mm
\textbf{Simulation of separable spatiotemporal multi-fidelity data.} \label{Simulation}
Let the spatial domain be a $4\times 4$ integer lattice
$\mathcal{S}=\{1,\dots,4\}\times\{1,\dots,4\}$ with $N_s=4^2=16$ stations,
and the temporal grid be $t_k\in[0,1]$ for $k=1,\dots,N_t$ with $N_t=15$ equispaced points
($\Delta t = 1/(N_t-1)$). Observations are ordered with \emph{time varying fastest} within each station.
We index the $N=N_s N_t$ space–time points by $i\!\leftrightarrow\!(\mathbf{s},t)$, where
$\mathbf{s}=(s_1,s_2)\in\mathcal{S}$ and $t\in\{t_1,\dots,t_{N_t}\}$. An example of the resulting station layout and train/test partition is shown in Figure~\ref{fig:sim_stations}.

\noindent
\hangindent=-5mm
\emph{Separable kernels.}
Define stationary RBF kernels in space and time,
\[
k_s^L(\mathbf{s},\mathbf{s}')=\sigma_L^2\exp\!\Big(-\tfrac12\,
\| \mathbf{s}-\mathbf{s}' \|_{\Lambda_L^{-1}}^2\Big),\qquad
k_t^L(t,t')=\sigma_L^2\exp\!\Big(-\tfrac12\,\tfrac{(t-t')^2}{\ell_{t,L}^2}\Big),
\]
and analogously $k_s^{\delta},k_t^{\delta}$ with variance $\sigma_\delta^2$ and
length-scales $(\Lambda_\delta,\ell_{t,\delta})$.
The full spatiotemporal covariances are \emph{separable}:
\[
K_L = K_s^L \circ K_t^L,\qquad
K_\delta = K_s^{\delta} \circ K_t^{\delta},
\]
where $K_s^L\in\mathbb{R}^{N_s\times N_s}$ and $K_t^L\in\mathbb{R}^{N_t\times N_t}$ are
the spatial and temporal Gram matrices on $\mathcal{S}$ and $\{t_k\}$, respectively,
and $\circ$ denotes the Hadamard (element-wise) product after Kronecker expansion to $\mathbb{R}^{N\times N}$:
$K_s^L\mapsto K_s^L\otimes \mathbf{1}_{N_t}$,\;
$K_t^L\mapsto \mathbf{1}_{N_s}\otimes K_t^L$.

\noindent
\hangindent=-5mm
\emph{Length-scales from target nearest-neighbor correlations.}
Given a desired one-step temporal correlation $c_t=0.8$ and nearest-neighbor spatial
correlations $c_s^{(L)}=0.8$ (low-fidelity latent) and $c_s^{(\delta)}=0.95$ (high-fidelity residual),
the corresponding length-scales are set so that
\[
\exp\!\Big(-\tfrac12(d/\ell)^2\Big)=c \quad\Longrightarrow\quad
\ell=\frac{d}{\sqrt{-2\log c}}.
\]
With unit spatial grid spacing $d=1$ and temporal step $d_t=\Delta t$,
this yields $\Lambda_L=\operatorname{diag}(\ell_{s1,L}^2,\ell_{s2,L}^2)$ with
$\ell_{s1,L}=\ell_{s2,L}=\frac{1}{\sqrt{-2\log 0.8}}$,
$\Lambda_\delta$ analogously with $0.95$, and
$\ell_{t,L}=\ell_{t,\delta}=\frac{\Delta t}{\sqrt{-2\log 0.8}}$.

\noindent
\hangindent=-5mm
\emph{Generative model.}
Let $\sigma_L^2=2.0$, $\sigma_\delta^2=0.8$, noise variances
$\sigma_{\varepsilon_L}^2=0.3$, $\sigma_{\varepsilon_\delta}^2=0.3$, and autoregressive coupling $\rho=0.6$.
We draw independent latent Gaussian processes on the $N$ stacked space–time points:
\[
\mathbf{d}_L \sim \mathcal{N}\!\big(\mathbf{0},\, K_L + \epsilon I\big),\qquad
\mathbf{d}_\delta \sim \mathcal{N}\!\big(\mathbf{0},\, K_\delta + \epsilon I\big),
\]
with numerical jitter $\epsilon=10^{-8}$.
Additive noises are drawn as
$\boldsymbol{\varepsilon}_L\sim\mathcal{N}(\mathbf{0},\sigma_{\varepsilon_L}^2 I)$
and $\boldsymbol{\varepsilon}_\delta\sim\mathcal{N}(\mathbf{0},\sigma_{\varepsilon_\delta}^2 I)$,
and the low- and high-fidelity signals are
\[
\mathbf{f}_L = \mathbf{d}_L + \boldsymbol{\varepsilon}_L,\qquad
\mathbf{f}_H = \rho\,\mathbf{f}_L + \big(\mathbf{d}_\delta + \boldsymbol{\varepsilon}_\delta\big).
\]
The outputs are then tabulated as $\{(\mathbf{s}_j,t_k,f_L(\mathbf{s}_j,t_k))\}$ and
$\{(\mathbf{s}_j,t_k,f_H(\mathbf{s}_j,t_k))\}$, sorted by station then time.

\noindent
\hangindent=-5mm
\emph{Train/test split by station.}
Let $\mathcal{I}_{\text{stn}}=\{1,\dots,N_s\}$ index stations.
For a given train fraction $\varphi\in[0,1]$ (default $\varphi=0.5$),
draw uniformly without replacement a training set
$\mathcal{I}_{\text{train}}\subset\mathcal{I}_{\text{stn}}$ with
$|\mathcal{I}_{\text{train}}|=\lfloor \varphi N_s\rfloor$, and define
$\mathcal{I}_{\text{test}}=\mathcal{I}_{\text{stn}}\setminus\mathcal{I}_{\text{train}}$.
All time points at stations in $\mathcal{I}_{\text{train}}$ form the training rows
and those in $\mathcal{I}_{\text{test}}$ form the test rows for both fidelities.
The function also returns the station–coordinate lookup
$\text{loc\_id}\mapsto(s_1,s_2)$.  Find the code for generating such data at the following GitHub \url{https://github.com/Pietrostat193/Robust-Multi-fidelity-Modelling/tree/main/Simulation/MatlabSimulations/PreliminaryExperiment}.

\begin{figure}
    \centering
    \includegraphics[width=1\textwidth]{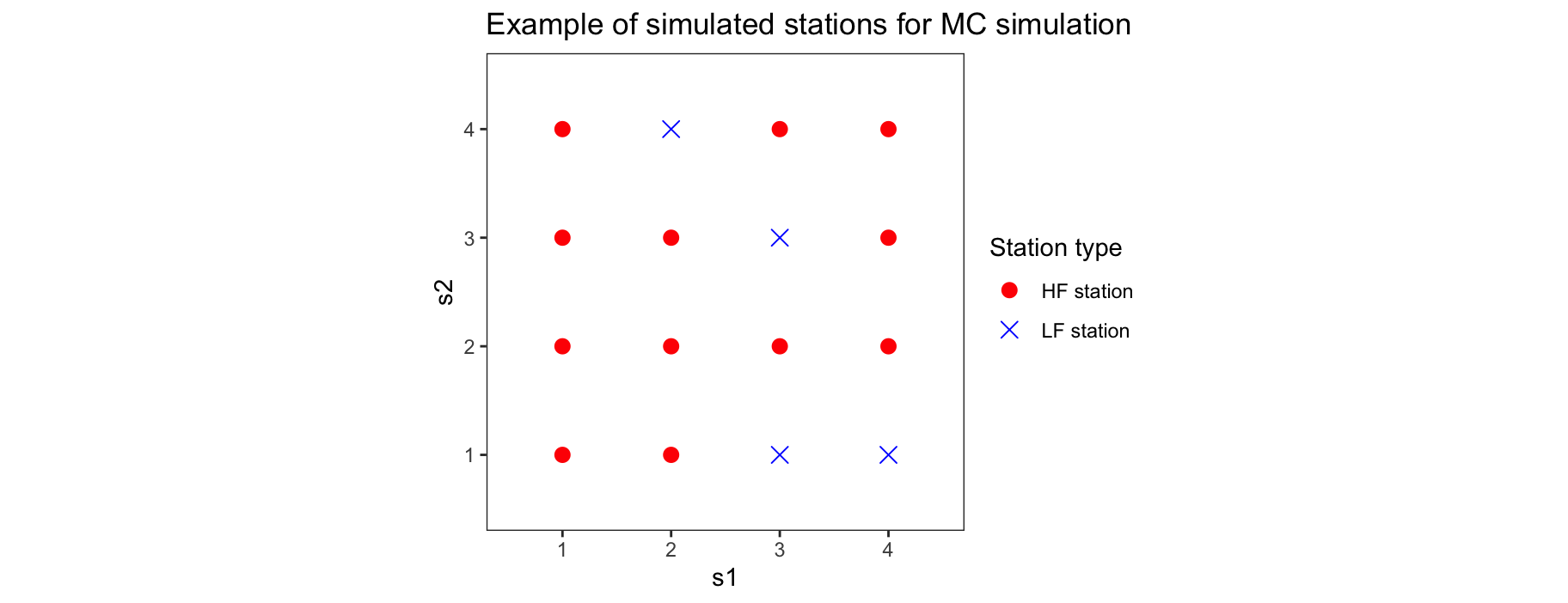}
    \caption{Example of the simulated spatial layout of stations on the $4\times 4$ integer lattice $\mathcal{S} = \{1,\dots,4\}^2$ used in the Monte Carlo experiments. Stations assigned to the training set (high- and low-fidelity observations) are shown in red, and stations reserved for testing (low-fidelity-only observations) are shown in blue. All time points at a given station share the same fidelity assignment.}
    \label{fig:sim_stations}
\end{figure}

\subsection{Proofs of the theoretical results}\label{sec:proofs}

\begin{proof}[Proof of Theorem \ref{thm:attenuation}]
Let $\|\cdot\|_{\Omega^{-1}}^2:=v^\top\Omega^{-1}v$. Given the observations $\mathbf{y}_L$ and the Gaussianity assumptions of the low- and high-fidelity process in \eqref{eq:mf_model}--\eqref{eq:joint_obs} (finite second-order moments), the fitted model implies $\mathbf{r}_H=\rho B\mathbf{r}_L+\boldsymbol{\eta}$.
The pseudo-true parameter minimises the population criterion
\[
Q(\rho)=\mathbb{E}\big\|\,\mathbf{r}_H-\rho B\mathbf{r}_L\,\big\|_{\Omega^{-1}}^2
=\mathbb{E}\!\left[(\mathbf{r}_H-\rho B\mathbf{r}_L)^\top \Omega^{-1}(\mathbf{r}_H-\rho B\mathbf{r}_L)\right],
\]
Expanding the quadratic form and differentiating with respect to $\rho$ yields
\begin{equation}
\frac{\partial Q(\rho)}{\partial \rho}
= -2\,\mathbb{E}\big[\mathbf{r}_H^\top \Omega^{-1} B \mathbf{r}_L\big]
   + 2\rho\,\mathbb{E}\big[\mathbf{r}_L^\top B^\top \Omega^{-1} B \mathbf{r}_L\big]. \label{eq:score2}
\end{equation}
Under the true data-generating process (contaminated low-fidelity process), we have $\mathbf{r}_L=\mathbf{f}_L+\mathbf{u}$) and $\mathbf{r}_H=\rho B\mathbf{f}_L+\boldsymbol{\eta}$. Centring the observed process and since $\mathbf{u}\perp(\mathbf{f}_L,\boldsymbol{\eta})$, we get that $\mathbb{E}[\mathbf{f}_L] = \mathbb{E}[\mathbf{u}]=\mathbf{0}$ and
\begin{align*}
\mathbb{E}\big[\mathbf{r}_H^\top \Omega^{-1} B \mathbf{r}_L\big]
&= \mathbb{E}\Big[(\rho B\mathbf{f}_L+\boldsymbol{\eta})^\top \Omega^{-1} B (\mathbf{f}_L+\mathbf{u})\Big]
 = \rho\,\mathbb{E}\big[\mathbf{f}_L^\top B^\top \Omega^{-1} B \mathbf{f}_L\big] \\
&= \rho\,\mathrm{tr}\!\big(B^\top \Omega^{-1} B\,C_L\big), \\[1mm]
\mathbb{E}\big[\mathbf{r}_L^\top B^\top \Omega^{-1} B \mathbf{r}_L\big]
&= \mathbb{E}\big[(\mathbf{f}_L+\mathbf{u})^\top B^\top \Omega^{-1} B (\mathbf{f}_L+\mathbf{u})\big] \\
&= \mathrm{tr}\!\big(B^\top \Omega^{-1} B\,(C_L+\Sigma_u)\big).
\end{align*}
Substituting into \eqref{eq:score2} and solving the equation yields the desired result. Finally, since $\Omega\succ 0$ and $B$ is fixed, $B^\top \Omega^{-1} B$ is positive semidefinite, and with $\Sigma_u\succeq 0$ the denominator in \eqref{eq:rho_star} is at least the numerator; thus $\kappa\in(0,1]$.
\end{proof}

\begin{proof}[Proof of Proposition \ref{prop:IF_outlier}]
If $\|\mathbf{u}\|\to\infty$, also $\|\mathbf{r}_L\|\to\infty$. From \eqref{eq:score},
\[
s_\rho(\mathbf{y})
= (\mathbf{r}_H - \rho B\mathbf{r}_L)^\top \Omega^{-1} B\mathbf{r}_L
= \mathbf{r}_H^\top \Omega^{-1} B\mathbf{r}_L
  - \rho\,\mathbf{r}_L^\top B^\top\Omega^{-1}B\,\mathbf{r}_L.
\]
The first term grows at most linearly with $\|\mathbf{r}_L\|$, while the second is quadratic. Since $B^\top\Omega^{-1}B$ is positive definite\footnote{Notice that, if $M=B^\top\Omega^{-1}B$ were only positive semidefinite (not definite), an outlier sequence with $\mathbf{r}_L$ asymptotically in $\mathrm{null}(M)$ would satisfy $\mathbf{r}_L^\top M \mathbf{r}_L=0$, so the quadratic term would not dominate and the score could grow at most linearly. This corresponds to LF directions completely uncoupled from HF via $B$. In our GP setting, however, $K_{LL}\succ0$, $\Sigma_{LL}\succ0$, and $\Omega\succ0$ imply $M\succ0$, so there is no nontrivial null space; we therefore do not consider this degenerate case.} under the covariance assumptions, $\mathbf{r}_L^\top B^\top\Omega^{-1}B\,\mathbf{r}_L \ge c\,\|\mathbf{r}_L\|^2$ for some $c>0$. Hence $|s_\rho(\mathbf{y})|\sim|\rho|\,\mathbf{r}_L^\top B^\top\Omega^{-1}B\,\mathbf{r}_L\to\infty$. Because the information constant $A=\mathbb{E}[\mathbf{r}_L^\top B^\top\Omega^{-1}B\,\mathbf{r}_L]$ is finite and positive, $|\mathrm{IF}(\mathbf{y};\rho)|=|A^{-1}s_\rho(\mathbf{y})|\to\infty$.
\end{proof}

\begin{proof}[Proof of Proposition \ref{prop:IF_shift}]
Let $\mathbf{r}_L=\mathbf{f}_L+\mathbf{b}$ with $\mathbf{b}$ defined as above. 
From \eqref{eq:score},
\[
s_\rho(\mathbf{y})
= \mathbf{r}_H^\top \Omega^{-1} B\mathbf{r}_L
  - \rho\,\mathbf{r}_L^\top B^\top\Omega^{-1}B\,\mathbf{r}_L.
\]
Substituting $\mathbf{r}_L=\mathbf{f}_L+\mathbf{b}$ and holding $\mathbf{r}_H$ fixed gives
\[
s_\rho(\mathbf{y})
= -\,\rho\,\mathbf{b}^\top B^\top\Omega^{-1}B\,\mathbf{b}
  - 2\rho\,\mathbf{f}_L^\top B^\top\Omega^{-1}B\,\mathbf{b}
  + O(1),
\]
where the first term is quadratic in $\mathbf{b}$, the second linear, and the remainder constant in $\Delta$.
Under the covariance assumptions ($K_{LL}\succ0$, $\Sigma_{LL}\succ0$, $\Omega\succ0$),
$M:=B^\top\Omega^{-1}B\succ0$, and hence
$\mathbf{b}^\top M\mathbf{b}\ge c\,\|\mathbf{b}\|^2$ for some $c>0$.
Since $\|\mathbf{b}\|^2 = |\Delta|^2|\mathcal{I}|$, the quadratic term dominates as $|\Delta|\to\infty$, giving
\[
|s_\rho(\mathbf{y})|
\sim |\rho|\,\Delta^2\,\mathbf{1}_{\mathcal{I}}^\top B^\top\Omega^{-1}B\,\mathbf{1}_{\mathcal{I}}\to\infty.
\]
Because the information is constant 
$A=\mathbb{E}[\mathbf{r}_L^\top B^\top\Omega^{-1}B\,\mathbf{r}_L]$ 
is finite and positive under the ideal model, it follows that 
$|\mathrm{IF}(\mathbf{y};\rho)|=|A^{-1}s_\rho(\mathbf{y})|\to\infty$.
\end{proof}

\begin{proof}[Proof of Theorem \ref{thm:bounded_IF_global}]
The Huber-based estimation corresponding to
\eqref{eq:ThetaHatHuber} can be written as
\[
\sum_{i=1}^{n_H} \psi_\delta(\tilde r_{H,i})\,
\frac{\partial \tilde r_{H,i}}{\partial \Theta}
= \mathbf{0},
\]
where $\psi_\delta(r)=\rho'_\delta(r)$ is the derivative (score function)
of the Huber loss. Since
$\tilde{\mathbf{r}}_H = W_{HH}^{1/2}\mathbf{r}_H$
and $\mathbf{r}_H = \mathbf{y}_H - \boldsymbol{\mu}_H(\Theta)$,
the score contribution of observation~$i$ is proportional to
$\psi_\delta(\tilde r_{H,i})$. 
By definition,
\[
|\psi_\delta(r)| =
\begin{cases}
|r|, & |r|\le \delta,\\
\delta, & |r|>\delta,
\end{cases}
\]
so each term is bounded by $\delta$ regardless of $r$. For the two shift types, we get
\begin{enumerate}
	\item When $\|\mathbf{u}\|\to\infty$ for a finite subset of LF indices, the high-fidelity residuals satisfy $\mathbf{r}_H=\rho B(\mathbf{f}_L+\mathbf{u})+\boldsymbol{\eta}
= \rho B\mathbf{f}_L+\rho B\mathbf{u}+\boldsymbol{\eta}$. After whitening, $\tilde{\mathbf{r}}_H=W_{HH}^{1/2}(\rho B\mathbf{u}+\cdots)$, so the affected components $\tilde r_{H,i}$ may diverge. However, $\psi_\delta(\tilde r_{H,i})$ saturates at $\pm\delta$, implying that the contribution of each outlier is bounded and, thus, the overall influence function remains finite.
\item For a shift $\mathbf{u}=\mathbf{b}$ with amplitude $|\Delta|\to\infty$ on an index set~$\mathcal{I}$, the induced residuals are $\mathbf{r}_H=\rho B\mathbf{b}+\text{(bounded terms)}$, and $\tilde r_{H,i}\asymp \Delta$ for $i\in\mathcal{I}$. Again, $\psi_\delta(\tilde r_{H,i})$ becomes constant in magnitude for $|\tilde r_{H,i}|>\delta$, so the corresponding score terms are bounded by $\delta$, and the influence function does not diverge.
\end{enumerate}

In both cases, the boundedness of $\psi_\delta(\cdot)$ ensures that
\[
\big\|\mathrm{IF}_H(\mathbf{y})\big\|
\;\le\; C_\delta < \infty,
\]
for some finite constant $C_\delta$ depending only on $\delta$ and the conditioning of $W_{HH}^{1/2}$ and $B$.
\end{proof}

\begin{proof}[Proof of Theorem \ref{thm:bounded_IF_global}]
The Huber-based estimation corresponding to
\eqref{eq:ThetaHatHuber} can be written as
\[
\sum_{i=1}^{n_H} \psi_\delta(\tilde r_{H,i})\,
\frac{\partial \tilde r_{H,i}}{\partial \Theta}
= \mathbf{0},
\]
where $\psi_\delta(r)=\rho'_\delta(r)$ is the derivative (score function) of the Huber loss. Since
$\tilde{\mathbf{r}}_H = W_{HH}^{1/2}\mathbf{r}_H$ and $\mathbf{r}_H = \mathbf{y}_H - \boldsymbol{\mu}_H(\Theta)$, the score contribution of observation~$i$ is proportional to $\psi_\delta(\tilde r_{H,i})$.  By definition,
\[
|\psi_\delta(r)| =
\begin{cases}
|r|, & |r|\le \delta,\\
\delta, & |r|>\delta,
\end{cases}
\]
so each term is bounded by $\delta$ regardless of $r$. For the two shift types, we get
\begin{enumerate}
	\item When $\|\mathbf{u}\|\to\infty$ for a finite subset of LF indices, the high-fidelity residuals satisfy $\mathbf{r}_H=\rho B(\mathbf{f}_L+\mathbf{u})+\boldsymbol{\eta} = \rho B\mathbf{f}_L+\rho B\mathbf{u}+\boldsymbol{\eta}$. After whitening, $\tilde{\mathbf{r}}_H=W_{HH}^{1/2}(\rho B\mathbf{u}+\cdots)$, so the affected components $\tilde r_{H,i}$ may diverge. However, $\psi_\delta(\tilde r_{H,i})$ saturates at $\pm\delta$, implying that the contribution of each outlier is bounded and, thus, the overall influence function remains finite.
\item For a shift $\mathbf{u}=\mathbf{b}$ with amplitude $|\Delta|\to\infty$ on an index set~$\mathcal{I}$, the induced residuals are $\mathbf{r}_H=\rho B\mathbf{b}+\text{(bounded terms)}$, and $\tilde r_{H,i}\asymp \Delta$ for $i\in\mathcal{I}$. Again, $\psi_\delta(\tilde r_{H,i})$ becomes constant in magnitude for $|\tilde r_{H,i}|>\delta$, so the corresponding score terms are bounded by $\delta$, and the influence function does not diverge.
\end{enumerate}

In both cases, the boundedness of $\psi_\delta(\cdot)$ ensures that
\[
\big\|\mathrm{IF}_H(\mathbf{y})\big\|
\;\le\; C_\delta < \infty,
\]
for some finite constant $C_\delta$ depending only on $\delta$ and the conditioning of $W_{HH}^{1/2}$ and $B$. 

To derive the bounds, notice that the influence function at the ideal model satisfies
\[
\mathrm{IF}_H(\mathbf y)\ =\ -\,J(\Theta_0)^{-1}\,S(\Theta_0;\mathbf y).
\]
Thus,
\begin{equation}
\|\mathrm{IF}_H(\mathbf y)\|\ \le\ \|J(\Theta_0)^{-1}\|_2\ \|S(\Theta_0;\mathbf y)\|.
\label{eq:IF-bound-step1}
\end{equation}
We now bound $\|S(\Theta_0;\mathbf y)\|_2$. By the definition of $S$ and using the Huber score, we get
\begin{equation}
	\|S(\Theta_0;\mathbf y)\|
\ \le\ \sum_{i=1}^{n_H}\big|\psi_\delta(\tilde r_{H,i}(\Theta_0))\big|\,\|g_i(\Theta_0)\| \le\ \delta\ \sum_{i=1}^{n_H}\|g_i(\Theta_0)\|. \label{eq:S-bound}
\end{equation}
Combining \eqref{eq:IF-bound-step1} and \eqref{eq:S-bound} yields \eqref{eq:Cdelta-general}.
\end{proof}

\begin{proof}[Proof of Lemma \ref{lem:alt-bounds-general}]
We first express the derivative of each whitened residual using the product rule. 
Recall that 
\(
\tilde{\mathbf r}_H(\Theta)=W_{HH}^{1/2}(\Theta)\,\mathbf r_H(\Theta)
\)
with 
\(
\mathbf r_H(\Theta)=\mathbf y_H-\boldsymbol{\mu}_H(\Theta).
\)
Differentiating componentwise gives
\[
g_i(\Theta)
= \frac{\partial \tilde r_{H,i}(\Theta)}{\partial \Theta^\top}
= e_i^\top\!\left(\frac{\partial W_{HH}^{1/2}}{\partial \Theta}\,\mathbf r_H(\Theta)
- W_{HH}^{1/2}(\Theta)\,\frac{\partial \boldsymbol{\mu}_H(\Theta)}{\partial \Theta}\right),
\]
where $e_i$ denotes the $i$th unit vector.

Applying the triangle inequality and the submultiplicativity of the spectral norm, 
\[
\|g_i(\Theta)\|
\le 
\Big\|\frac{\partial W_{HH}^{1/2}}{\partial \Theta}\Big\|
\|\mathbf r_H(\Theta)\|
+ 
\|W_{HH}^{1/2}(\Theta)\|
\Big\|\frac{\partial \boldsymbol{\mu}_H(\Theta)}{\partial \Theta}\Big\|.
\]
This inequality separates the contribution of (i) the change in the whitening operator
and (ii) the sensitivity of the model mean. 

Now, by assumption on the Lipschitz constants and the compactness of $\Xi$,
\[
\Big\|\frac{\partial W_{HH}^{1/2}}{\partial \Theta}\Big\| \le L_W, 
\quad
\Big\|\frac{\partial \boldsymbol{\mu}_H}{\partial \Theta}\Big\| \le L_\mu, 
\quad
\|W_{HH}^{1/2}\| \le \kappa_W,
\quad
\text{and}\quad
\|\mathbf r_H(\Theta)\| \le R.
\]
Substituting these uniform bounds gives
\[
\|g_i(\Theta)\| \le L_W R + \kappa_W L_\mu,
\quad\text{for all } i=1,\ldots,n_H.
\]

Finally, inserting this result into the general influence bound~\eqref{eq:Cdelta-general}
from Theorem~\ref{thm:bounded_IF_global} gives
\[
C_\delta \le \|J^{-1}\|_2\,\delta\,n_H(L_W R + \kappa_W L_\mu),
\]
which completes the proof.
\end{proof}

\begin{proof}[Proof of Lemma \ref{lem:alt-bounds-fixed}]
With $W_{HH}^{1/2}$ held fixed, the whitened residuals are
$\tilde{\mathbf r}_H(\Theta)=W_{HH}^{1/2}\,[\mathbf y_H-\boldsymbol{\mu}_H(\Theta)]$.
Differentiating gives
\[
g_i(\Theta)
= \frac{\partial \tilde r_{H,i}(\Theta)}{\partial \Theta^\top}
= -\,e_i^\top W_{HH}^{1/2}\,\frac{\partial \boldsymbol{\mu}_H(\Theta)}{\partial \Theta},
\]
so the Jacobian with rows $g_i(\Theta)^\top$ is
$G(\Theta)=-\,W_{HH}^{1/2}\,\frac{\partial \boldsymbol{\mu}_H}{\partial \Theta}$.

To bound $\sum_i \|g_i(\Theta)\|$, first apply Cauchy--Schwarz to the vector of row norms:
\[
\sum_{i=1}^{n_H}\|g_i(\Theta)\|
\;\le\; \sqrt{n_H}\,\Big(\sum_{i=1}^{n_H}\|g_i(\Theta)\|^2\Big)^{1/2}
\;=\; \sqrt{n_H}\,\|G(\Theta)\|_F.
\]
Next, use the submultiplicativity $\|AB\|_F \le \|A\|\,\|B\|_F$ with
$A=W_{HH}^{1/2}$ and $B=\partial \boldsymbol{\mu}_H/\partial \Theta$ to obtain
\[
\|G(\Theta)\|_F
= \big\|W_{HH}^{1/2}\,\tfrac{\partial \boldsymbol{\mu}_H}{\partial \Theta}\big\|_F
\;\le\; \|W_{HH}^{1/2}\|_2\,
\Big\|\frac{\partial \boldsymbol{\mu}_H}{\partial \Theta}\Big\|_F.
\]
Combining the two displays yields the claimed bound
\[
\sum_{i=1}^{n_H}\|g_i(\Theta)\|_2
\;\le\; \sqrt{n_H}\,\|W_{HH}^{1/2}\|_2\,
\Big\|\frac{\partial \boldsymbol{\mu}_H}{\partial \Theta}\Big\|_F,
\]
which is \eqref{eq:sumgi-fixed}. Substituting this into the general influence bound
\eqref{eq:Cdelta-general} from Theorem~\ref{thm:bounded_IF_global} immediately gives
\eqref{eq:Cdelta-fixedW-split}.
\end{proof}

\subsection{Additional empirical results}\label{sec:case_study}

\begin{table}
\centering
\caption{Descriptive PM\textsubscript{2.5} Statistics by ID}\label{tab:descriptive_full}
\begin{adjustbox}{max width=\textwidth}
\begin{tabular}{lrrrrrrrr}
\toprule
\textbf{Station ID} & \textbf{Count} & \textbf{Min} & \textbf{Max} & \textbf{Mean} & \textbf{Std.\ Error} & \textbf{95\% CI (Lower)} & \textbf{95\% CI (Upper)} \\
\midrule
1004 & 363 & 0.605 & 32.595 & 6.360 & 0.2695 & 5.832 & 6.888 \\
1070 & 325 & 0.248 & 30.943 & 4.624 & 0.2937 & 4.048 & 5.200 \\
1088 & 356 & 1.487 & 35.346 & 8.828 & 0.3165 & 8.207 & 9.448 \\
1134 & 365 & 1.603 & 40.858 & 12.847 & 0.3627 & 12.136 & 13.558 \\
1523 & 365 & 0.494 & 37.467 & 7.942 & 0.3493 & 7.257 & 8.626 \\
1577 & 348 & 1.612 & 244.986 & 9.610 & 0.7743 & 8.093 & 11.128 \\
1637 & 137 & 1.459 & 28.542 & 7.360 & 0.4471 & 6.484 & 8.237 \\
1682 & 364 & 0.867 & 44.103 & 9.922 & 0.4052 & 9.128 & 10.716 \\
1697 & 365 & 0.893 & 41.426 & 9.217 & 0.3503 & 8.531 & 9.904 \\
1735 & 365 & 0.403 & 14.229 & 1.085 & 0.0425 & 1.002 & 1.169 \\
239 & 365 & 1.947 & 38.418 & 10.372 & 0.3363 & 9.712 & 11.031 \\
2448 & 260 & 2.151 & 42.528 & 9.280 & 0.4053 & 8.486 & 10.075 \\
2470 & 365 & 2.522 & 40.299 & 10.883 & 0.3456 & 10.205 & 11.560 \\
2548 & 362 & 0.082 & 31.482 & 5.227 & 0.2751 & 4.687 & 5.766 \\
2558 & 364 & 2.115 & 29.058 & 8.561 & 0.2375 & 8.096 & 9.027 \\
2568 & 356 & 2.683 & 47.466 & 11.413 & 0.4262 & 10.578 & 12.249 \\
3022 & 318 & 1.085 & 48.024 & 10.546 & 0.4288 & 9.706 & 11.387 \\
317 & 351 & 0.798 & 41.998 & 8.178 & 0.3101 & 7.570 & 8.786 \\
3185 & 365 & 1.090 & 39.012 & 9.208 & 0.3625 & 8.498 & 9.919 \\
3274 & 365 & 1.008 & 33.389 & 8.122 & 0.2723 & 7.589 & 8.656 \\
3297 & 365 & 1.105 & 46.556 & 11.138 & 0.4470 & 10.262 & 12.014 \\
3503 & 365 & 0.741 & 36.048 & 7.908 & 0.3157 & 7.289 & 8.526 \\
353 & 364 & 1.804 & 40.771 & 9.997 & 0.3529 & 9.306 & 10.689 \\
3587 & 328 & 0.699 & 72.211 & 9.606 & 0.4486 & 8.727 & 10.485 \\
3629 & 358 & 0.990 & 39.059 & 8.901 & 0.3230 & 8.268 & 9.534 \\
3673 & 17 & 1.778 & 23.174 & 8.951 & 1.5954 & 5.824 & 12.078 \\
3825 & 111 & 0.782 & 27.220 & 6.818 & 0.5411 & 5.757 & 7.878 \\
3881 & 359 & 0.769 & 36.731 & 8.804 & 0.3872 & 8.045 & 9.563 \\
3997 & 344 & 0.933 & 44.956 & 10.665 & 0.4556 & 9.772 & 11.558 \\
4074 & 365 & 1.066 & 39.595 & 9.261 & 0.3649 & 8.546 & 9.977 \\
4288 & 320 & 2.637 & 999.900 & 38.034 & 7.0579 & 24.200 & 51.867 \\
4307 & 365 & 0.614 & 22.168 & 5.584 & 0.2052 & 5.181 & 5.986 \\
4309 & 364 & 0.817 & 350.158 & 9.076 & 1.2002 & 6.724 & 11.429 \\
4317 & 364 & 1.127 & 50.763 & 11.587 & 0.4397 & 10.725 & 12.449 \\
4377 & 360 & 0.908 & 38.697 & 8.170 & 0.3517 & 7.480 & 8.859 \\
543 & 196 & 0.818 & 28.223 & 5.336 & 0.2903 & 4.767 & 5.905 \\
591 & 346 & 0.616 & 33.852 & 7.954 & 0.3257 & 7.316 & 8.593 \\
836 & 365 & 1.149 & 38.769 & 9.278 & 0.3369 & 8.617 & 9.938 \\
938 & 338 & 0.389 & 15.903 & 3.632 & 0.1373 & 3.363 & 3.902 \\
980 & 362 & 1.803 & 34.194 & 10.307 & 0.2963 & 9.726 & 10.887 \\
STA.DE\_DEHH008 & 281 & 1.500 & 40.307 & 10.639 & 0.4160 & 9.823 & 11.454 \\
STA.DE\_DEHH015 & 270 & 2.144 & 35.325 & 10.114 & 0.3843 & 9.361 & 10.867 \\
STA.DE\_DEHH059 & 289 & 2.239 & 40.320 & 11.702 & 0.3905 & 10.937 & 12.468 \\
STA.DE\_DEHH068 & 324 & 1.500 & 38.333 & 10.556 & 0.3753 & 9.821 & 11.292 \\
\bottomrule
\end{tabular}
\end{adjustbox}
\end{table}

\end{document}